\documentclass[a4paper, UKenglish, cleveref, autoref, thm-restate]{lipics-v2021}
\usepackage{graphicx}
\usepackage{hyperref}
\usepackage{subfiles}
\usepackage{amsmath}

\usepackage{comment}
\usepackage[dvipsnames]{xcolor}
\usepackage{boxedminipage}
\usepackage[ruled,vlined]{algorithm2e}
\usepackage[nocompress]{cite}
\usepackage{lineno}
\usepackage[capitalise, noabbrev]{cleveref}
\usepackage{csquotes}

\newcommand{\R}{\mathbb{R}}
\newcommand{\N}{\mathbb{N}}
\newcommand{\Z}{\mathbb{Z}}
\newcommand{\set}[1]{\{ #1 \}}
\newcommand{\fromto}[2]{\set{#1, \ldots, #2}}
\newcommand{\bigO}{O}
\newcommand{\dotunion}{\mathbin{\dot{\cup}}}
\newcommand{\nequiv}{\not\equiv}

\newcommand{\EM}{\textsc{Em}}
\newcommand{\PIT}{\textsc{Pit}}
\newcommand{\BCPM}{\textsc{Bcpm}}
\newcommand{\CPM}{\textsc{Cpm}}
\newcommand{\pshort}{\textsc{Pshort}}
\newcommand{\lmod}{\textsc{L-Mod}}
\newcommand{\dlmod}{\textsc{Dual-L-Mod}}
\newcommand{\Mmin}{M_\text{min}}

\newcommand{\local}{\textsc{Local}}
\newcommand{\Cyc}{\mathcal{C}}

\newcommand{\Poly}{\mathsf{P}}

\newcommand{\BPP}{\mathsf{BPP}}
\newcommand{\RP}{\mathsf{RP}}
\newcommand{\CoRP}{\mathsf{CoRP}}
\newcommand{\RNC}{\mathsf{RNC}}

\DeclareMathOperator{\Ram}{\text{Ram}}
\DeclareMathOperator{\dist}{dist}
\DeclareMathOperator{\Nb}{\mathcal{N}}
\newcommand{\new}{\text{new}}

\newcommand{\symdif}{\mathbin{\triangle}}
\newcommand{\rank}{\text{rank}}

\pdfoutput=1
\hideLIPIcs 

\title{On the Exact Matching Problem in Dense Graphs}

\author{Nicolas {El Maalouly}}{Department of Computer Science, ETH Zurich, Switzerland }{nicolas.elmaalouly@inf.ethz.ch}{0000-0002-1037-0203}{}
\author{Sebastian {Haslebacher}}{Department of Computer Science, ETH Zurich, Switzerland }{sebastian.haslebacher@inf.ethz.ch}{0000-0003-3988-3325}{}
\author{Lasse {Wulf}}{Institute of Discrete Mathematics, TU Graz, Austria }{wulf@math.tugraz.at}{0000-0001-7139-4092}{Supported by the Austrian Science Fund (FWF): W1230.}

\authorrunning{N. El Maalouly, S. Haslebacher, L. Wulf}

\Copyright{Nicolas El Maalouly, Sebastian Haslebacher, Lasse Wulf} 



\nolinenumbers

\begin{document}
\maketitle

\begin{abstract}
In the Exact Matching problem, we are given a graph whose edges are colored red or blue and the task is to decide for a given integer $k$, if there is a perfect matching with exactly $k$ red edges. Since 1987 it is known that the Exact Matching Problem can be solved in randomized polynomial time. Despite numerous efforts, it is still not known today whether a deterministic polynomial-time algorithm exists as well.
In this paper, we make substantial progress by solving the problem for a multitude of different classes of dense graphs.
We solve the Exact Matching problem in deterministic polynomial time for complete $r$-partite graphs, for unit interval graphs, for bipartite unit interval graphs, for graphs of bounded neighborhood diversity, for chain graphs, and for graphs without a complete bipartite $t$-hole. We solve the problem in quasi-polynomial time for Erd\H{o}s-Rényi random graphs $G(n, 1/2)$. We also reprove an earlier result for bounded independence number/bipartite independence number. We use two main tools to obtain these results: A local search algorithm as well as a generalization of an earlier result by Karzanov.
\end{abstract}

\newpage 

\section{Introduction}

A fundamental problem in computer science is the question whether a randomized algorithm has any sort of advantage over a deterministic algorithm.
In particular, theoretical computer scientists are concerned with the question: $\Poly = \BPP$? 
Here, $\Poly$ contains decision problems that can be solved deterministically in polynomial-time, while $\BPP$ contains decision problems that can be solved with randomized algorithms in polynomial-time (under a two-sided, bounded error probability \cite{arora2009computational}). 
One can also define the classes $\RP \subseteq \BPP$ and $\CoRP \subseteq \BPP$ of randomized polynomial-time algorithms with one-sided error probability (the difference between the two classes is the side of the error). Nowadays, experts in complexity theory believe that $\Poly = \RP = \CoRP = \BPP$, i.e.\ it is believed that randomness does not offer any sort of advantage for the task of solving a problem in polynomial time. The reason for this belief are deep connections between complexity theory, circuit lower bounds, and pseudorandom generators \cite{kabanets2003derandomizing,impagliazzo1997p,arora2009computational}.

While it would be intriguing to attack the conjecture $\Poly = \BPP$ directly, it seems very hard to make direct progress in this way. In particular, $\Poly = \BPP$ would imply deterministic algorithms for \emph{all} problems which can be solved with randomness. A more humble approach one can take is to look for one specific problem, where the research community knows a randomized, but no deterministic algorithm, and try to find a deterministic algorithm for this specific problem. Every single of these results can be seen as further evidence towards $\Poly = \BPP$. One famous example of such a \enquote{derandomization} is the deterministic algorithm for primality testing by Agrawal, Kayal and Saxena\@~\cite{agrawal2004primes} from 2002. 

Quite interestingly, we only know of a handful of problems where a randomized but no deterministic polynomial-time algorithm is known. This paper is concerned with one of these examples, the \emph{Exact Matching problem} (\EM). Given an integer $k$ and a simple graph $G$ together with a coloring of its edges in red or blue, \EM\ is the problem of deciding whether $G$ has a perfect matching with exactly $k$ red edges. 
\EM\ was introduced by Papadimitriou and Yannakakis~\cite{papadimitriou1982complexity} back in 1982. 
Not too long after its introduction, in 1987, Mulmuley et al.\@~\cite{mulmuley1987matching} showed that \EM\ can be solved in randomized polynomial-time. Despite the original problem being from 1982, and in spite of multiple applications of \EM\ in different areas (see next paragraph), it is still not known today, if a deterministic polynomial-time algorithm exists.

Another interesting aspect of \EM\ is its connection to \emph{polynomial identity testing} (\PIT). \PIT\ is another one of the rare problems in $\BPP$ for which we still do not know any deterministic polynomial-time algorithm.
Given a multivariate polynomial described by an algebraic circuit, \PIT\ is the problem of deciding whether the polynomial is identically equal to zero or not. Using the well-known Schwartz-Zippel Lemma (named after Schwartz~\cite{schwartz1980fast} and Zippel~\cite{zippel1979probabilistic} who discovered it in the eighties), it is clear that \PIT\ belongs to $\CoRP$. Therefore, under the conjecture $\CoRP = \Poly$, there should be a deterministic polynomial-time algorithm for \PIT. However, Kabanets and Impagliazzo~\cite{kabanets2003derandomizing} 
provided strong evidence that derandomizing \PIT\ might be notoriously hard, since it would imply proving circuit lower bounds. The known randomized algorithm for \EM\ uses \PIT\ as a subroutine on a slightly modified Tutte matrix of the given graph. Alternatively, one can substitute the use of \PIT\ with the famous Isolation Lemma due to Mulmuley et al.\@~\cite{mulmuley1987matching}.
Both approaches lead to randomized polynomial-time algorithms for \EM\ and show that \EM\ is contained in the class $\RP$.

\paragraph*{History of Exact Matching}
\label{section:history_EM}

We have already established that \EM\ should belong to $\Poly$ if we believe the conjecture $\Poly = \RP = \BPP$. However, the best deterministic algorithm to date takes exponential time. 
This is especially astonishing knowing that \EM\ was introduced by Papadimitriou and Yannakakis~\cite{papadimitriou1982complexity} back in 1982.
A few years later, in 1987, Mulmuley et al.\@~\cite{mulmuley1987matching} showed that \EM\ can be solved in randomized polynomial-time in their famous paper that also introduced the Isolation Lemma. In fact, their algorithm additionally allows for a high degree of parallelism i.e.\ they proved that \EM\ belongs to $\RNC$ (and hence also to $\RP$ and $\BPP$). $\RNC$ is defined as the class of decision problems allowing an algorithm running in polylogarithmic time
using polynomially many parallel processors, while having additional access to randomness (we refer the interested reader to~\cite[Chapter 12]{complexitybook} for a formal definition). This means that if we allow for randomness, \EM\ can be solved efficiently even in parallel, while the best known deterministic algorithm requires exponential time.

In the same year 1987,
Karzanov~\cite{karzanov1987maximum} gave a precise characterization of the solution landscape of \EM\ in complete and complete bipartite graphs. His characterization also implies deterministic polynomial-time algorithms for \EM\ restricted to those graph classes. Several articles later appeared \cite{yi2002matchings,geerdes,gurjar2017exact}, simplifying and restructuring those results. 

\EM\ is known to admit efficient deterministic algorithms on some other restricted graph classes as well: With standard dynamic programming techniques, \EM\ can be solved in polynomial-time on graphs of bounded tree-width~\cite{elmaalouly2022exacttopk, vardi2022quantum}. 
Moreover, derandomization results exist for $K_{3, 3}$-minor free graphs~\cite{vazirani1989nc, yuster2012almost} and graphs of bounded genus~\cite{galluccio1999theory}. These works make use of so-called Pfaffian orientations. 
Besides solving \EM\ on restricted graph classes, some prior work has also focused on solving \EM\ approximately. Yuster~\cite{yuster2012almost} proved that in a YES-instance, we can always find an almost exact matching in deterministic polynomial-time (an almost exact matching is a matching with exactly $k$ red edges that fails to cover only two vertices). 

This completes a summary of the history of the \EM\ problem until two years ago. With so little progress, one might wonder if the community has lost interest in, or forgot about the problem. However, over the last decade alone, we have seen the problem appear in the literature from several areas. This includes budgeted, color bounded, or constrained matching~\cite{berger2011budgeted,mastrolilli2012constrained,mastrolilli2014bi,stamoulis2014approximation,kelk2019integrality}, multicriteria optimization~\cite{grandoni2010optimization}, matroid intersection for represented matroids~\cite{camerini1992random}, binary linear equation systems~\cite{arvind2016solving}, recoverable robust assignment~\cite{fischer2020investigation}, or planarizing gadgets for perfect matchings~\cite{gurjar2012planarizing}. In many of these papers, a full derandomization of \EM\ would also derandomize some or all of the results of the paper. \EM\ also appeared as an interesting open problem in the seminal work on the parallel computation complexity of the matching problem~\cite{svensson2017matching}, which might be partly responsible for the increase in attention that the problem has received recently.

\paragraph*{Recent progress}

As mentioned above, until recently, \EM\ was only solved for a handful of graph classes. This is even more extreme in the case of dense graphs where it was only solved on complete and complete bipartite graphs. In 2022, El Maalouly and 
Steiner~\cite{elmaalouly2022exact} finally made progress on this side by
showing that \EM\ can be solved on graphs of
bounded \emph{independence number} and bipartite graphs of bounded \emph{bipartite independence number}. Here, the independence number of a graph $G$ is defined as the largest number $\alpha$ such that $G$ contains an \emph{independent set} of size $\alpha$. The bipartite independence number of a bipartite graph $G$ equipped with a bipartition of its vertices is defined as the largest number $\beta$ such that $G$ contains a \emph{balanced independent set} of size $2\beta$, i.e., an independent set using exactly $\beta$ vertices from both color classes. This generalizes previous results for complete and complete bipartite graphs, which correspond to the special cases $\alpha=1$ and $\beta=0$. The authors also conjectured that counting perfect matchings is \#P-hard for this class of graphs. This conjecture was later proven in \cite{elmaalouly2022hard} already for $\alpha = 2$ or $\beta = 3$. 
This makes them the first classes of graphs where \EM\ can be solved, even though counting perfect matchings is \#P-hard.
This work was later extended
to an FPT-algorithm on bipartite graphs parameterized by the bipartite independence number~\cite{elmaalouly2022exact}. 

There has also been a recent interest in approximation algorithms for \EM. Such approximation algorithms have been developed for the closely related \emph{budgeted matching} problem, where sophisticated methods were used to achieve a PTAS \cite{berger2011budgeted} and, more recently, an efficient PTAS \cite{budgetMatchingEPTAS}. These methods however do not guarantee to return a perfect matching (but note that a deterministic FPTAS for budgeted matching would imply a deterministic polynomial-time algorithm for \EM~\cite{berger2011budgeted}). In \cite{elmaalouly2022exacttopk,durr2023approximation}  it is argued that relaxing the perfect matching constraint takes away most of the difficulty of the problem. In contrast, the aim of the recent work has been to keep the perfect matching constraint and relax the requirement on the number of red edges. The first such result was given in \cite{elmaalouly2022exact}, where it was shown that in a bipartite graph we can always find a perfect matching with at least $0.5k$ and at most $1.5k$ red edges in deterministic polynomial-time. This represents a two-sided approximation for the problem. 
Shortly after, \cite{durr2023approximation} studied the surprisingly much more difficult problem of getting a one-sided approximation and presented a 3-approximation in that setting (i.e. an algorithm that outputs a perfect matching with at least $k/3$ and at most $k$ red edges), relying on a newly defined concept of graph rigidity.

Another relaxation of the problem is to consider only modular constraints on the number of red edges, e.g., requiring the output perfect matching to have an odd number of red edges. In the case of bipartite graphs, the problem can be solved using the more general result of \cite{artmann2017strongly} on network matrices. This does not work for general graphs, for which the problem was solved in \cite{Maalouly_Steiner_Wulf} with a different approach relying on a deep result by Lov\'{a}sz~\cite{lovasz} on the \emph{linear hull} of perfect matchings. The problem remains open for other congruency constraints, e.g., requiring the output to have $(r\ \mathrm{mod}\ p)$ red edges for some integers $r$ and $p$. The latter problem has been used by \cite{nagele2023advances} as a building block in an algorithm for a special class of integer programs having a constraint matrix with bounded subdeterminants. This means that a deterministic algorithm for this special case of \EM\ would also derandomize the algorithm of \cite{nagele2023advances}.

In \cite{elmaalouly2022exacttopk}, the Top-$k$ Perfect Matching problem is introduced, where the input is a weighted graph and the goal is to find a PM that maximizes the weight of the $k$ heaviest edges in the matching. In combination with the result from \cite{maalouly2022exact}, the problem is shown to be polynomially equivalent to \EM\ when the input weights are polynomially bounded. Several approximation and FPT algorithms were also developed.

Another recent line of work follows a polyhedral approach to understand the differences between finding a perfect matching and \EM ~\cite{jia2023exact}. In particular, the authors show exponential extension complexity for the bipartite exact matching polytope. This stands in contrast with the bipartite perfect matching polytope whose vertices are all integral~\cite{edmonds1965paths}.

Finally, \cite{vardi2022quantum} studies some generalizations of \EM\ to matching problems with vertex color constraints and shows an interesting connection to quantum computing.

\paragraph*{Our contribution}
In this paper, we study \EM\ on dense graph classes. We are able to solve \EM\ in deterministic polynomial time on many different classes of dense graphs, which before could only be handled by a randomized algorithm. In order to achieve this result, we use two key techniques: First, a local search algorithm, second, a generalization of Karzanov's \cite{karzanov1987maximum} theorem.
With the first technique, the local search algorithm, we obtain the following results: (For a formal definition of all the graph classes listed, as well as a motivation for why we consider exactly these classes, we refer the reader to \cref{subsec:graph_class_definitions}.)

\begin{itemize}
\item There is a deterministic $n^{\bigO(1)}$ time algorithm for {\EM} on complete $r$-partite graphs for all $r \geq 1$. The constant in the exponent is independent of $r$. This is an extension of the special cases $r=n$ and $r = 2$, which correspond to the cases of complete and complete bipartite graphs \cite{karzanov1987maximum} already known in 1987. 
\item  There is a deterministic $n^{\bigO(1)}$ time algorithm for {\EM} on graphs of bounded neighborhood diversity $d = O(1)$. The neighborhood diversity is a parameter popular in the area of parameterized complexity \cite{lampis2012algorithmic}. 
\item  There is a deterministic $n^{\bigO(1)}$ time algorithm for {\EM} on graphs $G$ which have no complete bipartite $t$-hole (i.e.\ $K_{t,t} \not\subseteq \overline{G}$) with $t = O(1)$. 
\item There is a deterministic $n^{\bigO(\log^{12}(n) p^{-12})}$ time algorithm for {\EM} on the random graph $G(n,p)$. By this, we mean the following: We say an algorithm is \emph{correct} for a graph $G$, if for all possible red-blue edge colorings of $G$ and all possible $k$, the algorithm correctly solves {\EM} on that input. We show that there is a deterministic algorithm $\mathcal{A}$, which always halts in $n^{\bigO(\log^{12}(n) p^{-12})}$ steps, and if $G$ is sampled from the distribution $G(n,p)$, then with high probability $\mathcal{A}$ is correct for $G$.
 As a special case, we obtain a quasi-polynomial algorithm for $G(n,1/2)$. We are the first authors to consider {\EM} from the perspective of random graphs.

\item As a special case, our main theorem contains a re-proof of the two main results of \cite{elmaalouly2022exact}, showing that there is a deterministic $n^{O(1)}$ algorithm for \EM\ on graphs of bounded independence number/ bip.\ graphs of bounded bip.\ independence number.  Our result is therefore a large generalization of this earlier result and puts it into the bigger context of local search.
\end{itemize} 

We also identify a certain graph property, which we call the \emph{path-shortening property}. Graphs which are very dense and structured are candidates to examine for this property. 
Our main theorem is that for every graph with the path-shortening property, a local search approach can be used to correctly solve the Exact Matching Problem. 
In fact, all the examples above follow from our main theorem. We remark that our local search algorithm is very simple, only the proof of its correctness is quite involved. 
The main idea of the observation is that in graphs with the path-shortening property, strong locality statements about the set of all perfect matchings can be made. Details are presented in \cref{sec:local_search}. 

While the local search approach allows us to tackle several new graph classes, we still notice that it fails even on some very dense and structured graph classes. In particular, we are interested in graph classes which are related to the problem of counting perfect matchings (for example, Okamoto et al.\ \cite{okamoto2009counting} list chordal, interval, unit interval, bipartite chordal, bipartite interval, and chain graphs among others). We highlight one example, the case of so-called chain graphs, where our local search fails.

This failure inspires us to seek other methods to understand the \EM\ problem on dense graphs and leads us to consider our second key technique. We call this technique \emph{Karzanov's property}, as it is a generalization of the result by Karzanov~\cite{karzanov1987maximum}. 
We show that several graph classes have Karzanov's property, including classes where our local search algorithm fails. We introduce a related property, which we call the \emph{chord property}. We introduce a novel binary-search like procedure, which gives us both an efficient algorithm for \EM\ on these graph classes, as well as a characterization of their solution landscape. 

Finally, we are also able to identify some graph classes, where Karzanov's property almost holds. We call this the \emph{weak Karzanov's property}. For those graph classes, we are not able to solve \EM\ deterministically, but we are at least able to show that one can always find a PM with either $k$ or $k-1$ red edges. Hence we can come very close to solving the problem. These graph classes are therefore obvious candidates to attack next in the effort of derandomizing \EM. In summary, we obtain the following.

\begin{itemize}
\item There is a deterministic $n^{\bigO(1)}$ time algorithm for {\EM} on chain graphs, unit interval graphs, bipartite unit interval graphs and complete $r$-partite graphs for all $r \geq 1$. The solution landscape for these graph classes can also be characterized by the perfect matchings of maximum and minimum number of red edges with a given parity.
\item There is a deterministic $n^{\bigO(1)}$ time algorithm on interval graphs, bipartite interval graphs, strongly chordal graphs and bipartite chordal graphs, that outputs a perfect matching with either $k-1$ or $k$ red edges or deduces that the answer of the given {\EM}-instance is "No".
\end{itemize} 

\paragraph*{Organization of the paper}

In the following we start with some preliminaries in \Cref{sec:preliminaries}. Then in \Cref{sec:local_search} we introduce our local search algorithm, the main ideas behind its correctness and its limitations. In \Cref{sec:karzanov} we introduce Karzanov's property, as well as Karzanov's weak property and discuss the main ideas behind their utility and limitations. In \Cref{sec:proof-main-thm} we give the detailed proofs missing from \Cref{sec:local_search} and show that the local search algorithm works for several graph classes. In \Cref{section:proofs_for_karzanov} we give the detailed proofs missing from \Cref{sec:karzanov} and show that several graph classes satisfy Karzanov's property while some others only satisfy Karzanov's weak property. 

\section{Preliminaries and Problem Definition}\label{sec:preliminaries}

All graphs in this paper are undirected and simple. For a graph $G = (V,E)$, we denote by $V(G) := V$ its vertex set and by $E(G) := E$ its edge set. We usually use the letters $n, m$ to denote $n := |V|$ and $m := |E|$. In this paper, paths and cycles are always simple (i.e. no vertex is repeated). In order to simplify the notation, we identify paths and cycles with their edge sets. Any reference to their vertices will be made explicit. The \emph{neighborhood} $N(v)$ of a vertex $v$ is the set of all vertices adjacent to $v$. A \emph{colored graph} in this paper is a graph where every edge has exactly one of two colors, i.e.\ a tuple $(G, c)$ with $c : E(G) \to \set{\text{red}, \text{blue}}$. For a subset $F \subseteq E$ of edges, we denote by $R(F) := \set{e \in F \mid e\text{ is red}}$ its set of red edges and by $r(F) := |R(F)|$ its number of red edges. Analogously, we define $B(F)$ and $b(F)$ for blue edges. A \emph{matching} of $G$ is a set $M \subseteq E$ of edges, which touches every vertex at most once. A \emph{perfect matching} (abbreviated PM) is a matching $M$ which touches every vertex exactly once. An edge $e$ is called \emph{matching}, if $e \in M$ and \emph{non-matching} otherwise. The \emph{Exact Matching Problem} is formally defined as follows.
\begin{quote}
\textbf{Problem} $\EM$\\
\textbf{Input:} Colored graph $(G, c)$, integer $k \geq 0$.\\
\textbf{Question:} Is there a perfect matching $M$ in $G$ such that $r(M) = k$? 
\end{quote}
The \emph{symmetric difference} $A \symdif B$ of two sets $A$ and $B$ is $A \cup B \setminus (A \cap B)$. Let $M$ be a PM. An \emph{$M$-alternating cycle}, or simply an \emph{alternating cycle} (if $M$ is clear from context), is a cycle which alternates between edges in $M$ and edges not in $M$. An \emph{alternating path} is defined analogously. If $M_1,M_2$ are PMs, it is well known that $D := M_1 \symdif M_2$ is a vertex-disjoint union of alternating cycles. Since $\symdif$ behaves like addition mod 2, we also have $M_2 = M_1 \symdif D$ and $M_1 = M_2 \symdif D$. In this paper, we try to follow the convention that the letter $C$ denotes a single cycle and the letter $D$ denotes a vertex-disjoint union of one or more cycles.

\subsection{Definition of Graph Classes}
\label{subsec:graph_class_definitions}
Throughout the paper, we show how to solve \EM\ on various classes of dense graphs. In this subsection, we properly define all graph classes used.

The motivation to consider exactly those classes comes from different sources. 
Some of the classes considered are direct generalizations of classes, where it was previously known that \EM\ can be solved. Other classes are generally well-known.
For the remaining classes, we regard them as interesting, because they appear in the context of counting the number of perfect matchings. (For example, in their paper about counting perfect matchings, Okamoto et al.\ \cite{okamoto2009counting} list chordal, interval, unit interval, bipartite chordal, bipartite interval, and chain graphs among others.) The reason for this is, that if it is \#P-hard to count the number of perfect matchings, then the Pfaffian derandomization method used in \cite{vazirani1989nc, yuster2012almost,galluccio1999theory} is unlikely to work (compare \cite{maalouly2022exact} for details). We also pay special attention to bipartite graphs, since we expect \EM\ to be easier to tackle if the graph is bipartite.

A graph is \emph{complete $r$-partite}, if the vertex set can be partititioned into $r$ parts $V_1,\dots,V_r$ such that inside each part there are no edges, and between two different parts, there are all the possible edges. 
The case $r=n$ corresponds to the complete graph, while $r=2$ corresponds to the complete bipartite graph. 
A generalization of complete $r$-partite graphs, if $r = O(1)$, are graphs of \emph{bounded neighborhood diversity}, a parameter coming from parameterized complexity \cite{lampis2012algorithmic}. 
A graph has neighborhood diversity $d$, if $V(G)$ can be partitioned into $d$ parts $V_1\dots,V_d$, such that between every two parts $V_i, V_j$ with $i \neq j$, there are either no edges, or all the possible edges, and every part itself induces either a complete or an empty graph. 

The Erd\H{o}s-Rényi random graph $G(n,p)$ is the random graph on $n$ vertices, where every edge appears with probability $p$ independently \cite{erdHos1960evolution}. It is very well-studied and has a rich history.

The remaining definitions in this subsection are motivated by \cite{okamoto2009counting}. Furthermore, many of these graph classes are extensively studied in algorithmic graph theory \cite{golumbic2004algorithmic}.
A graph $G = (V, E)$ is an \emph{interval graph} if there exists a mapping 
        $I : V \rightarrow \{[a, b] \subseteq \R \mid a \leq b\}$ such that 
        $\{u, v\} \in E \iff I(u) \cap I(v) \neq \emptyset$ holds for all distinct $u, v \in V$. If additionally $I(v)$ is a unit interval for all vertices $v$, then $G$ is called a \emph{unit interval graph}.

 We may consider bipartite versions of interval graphs the following way: A bipartite graph $G = (X \dotunion Y, E)$ is a \emph{bipartite interval graph} if 
    there exists a mapping 
        $I : X \dotunion Y \rightarrow \{[a, b] \subseteq \R \mid a \leq b \}$
    such that 
        $\{x, y\} \in E \iff I(x) \cap I(y) \neq \emptyset$
    holds for all $x \in X, y \in Y$. If additionally $I(v)$ is a unit interval for all vertices $v$, then $G$ is called a \emph{bipartite unit interval graph}.
Note that by this definition, if two vertices $x,y \in X$ are in the same color class of the bipartition, there is no edge between $x$ and $y$ even if the intervals $I(x)$ and $I(y)$ intersect.
    
Interval graphs are \emph{strongly chordal}.  A graph $G$ is called strongly chordal if every cycle of length at least $4$ 
    admits a chord and every even cycle of length
    at least $6$ admits an odd chord (i.e. a chord that splits the cycle into two odd length paths).

Bipartite interval graphs are \emph{bipartite chordal}.  A bipartite graph $G$ is bipartite chordal if and only if every cycle of (necessarily even) 
    length at least $6$ admits a chord.

Finally, we consider a special case of bipartite interval graphs, so-called \emph{chain graphs}. A bipartite graph $G = (X \dotunion Y, E)$ is a chain graph if and only if its vertices can be relabeled as 
    $x_1, \dots, x_{|X|} \in X$ and $y_1, \dots, y_{|Y|} \in Y$
    such that $N(x_i) \subseteq N(x_{i + 1})$ and $N(y_j) \subseteq N(y_{j + 1})$ 
    hold for all $1 \leq i < |X| $ and $1 \leq j < |Y|$.

\section{Local Search}
\label{sec:local_search}

As of course is well known, the central idea behind a local search algorithm is to only examine solutions close to the current solution at every step. Hence we require a notion of distance. For our purpose, this notion is as follows.
\begin{definition}
Let $(G, c)$ be a colored graph and $M_1, M_2 \subseteq E(G)$ be two PMs. The distance between $M_1, M_2$ is
\[\dist(M_1, M_2) := \min\set{r(M_1\symdif M_2),\ b(M_1 \symdif M_2).}\]
For an integer $s \geq 0$, the $s$-neighborhood of a PM $M$ is
\[ \Nb_s(M) := \set{M' \subseteq E(G) \mid M' \text{ is a PM},\ \dist(M,M') \leq s}.\]
\end{definition}
Note that $\dist(M_1, M_2) = \min\set{|R(M_1) \symdif R(M_2)|,\ |B(M_1) \symdif B(M_2)|}$. In other words, two PMs have small distance if and only if their two sets of red edges are almost the same, or their two sets of blue edges are almost the same. For example, if two PMs have the same set of red edges, i.e.\ $R(M_1) = R(M_2)$, then $\dist(M_1, M_2) = 0$, even if their set of blue edges is completely different.

Observe that as a consequence of this definition, for a fixed PM $M$ even the 0-neighborhood $\Nb_0(M)$ may have exponential size in $n$. 
This is a problem for us: how can we perform local search, if the size of the neighborhood is exponential? Fortunately, there is a fix: 
We do not need to know the complete neighborhood of $M$, all we need to know is which values of $r(M')$ are possible to achieve in the neighborhood, i.e.\ the set of all $k'$ such that there exists a PM $M'$ in the neighborhood with $r(M') = k'$. The following lemma states that this information can be computed efficiently. The idea is to guess either the set $R(M')$ or the set $B(M')$ and see if the guess can be completed to a PM using only edges of the opposite color.

\begin{lemma}
\label{obs:guess-difference}
Assume we are given a PM $M$ in a colored graph, and an integer $s \geq 0$. There is an algorithm which runs in $\bigO(m^{s+3})$ time and computes the set $\set{k' \in \N \mid \exists M' \in \Nb_s(M), r(M') = k'}$ and for each $k'$ in this set outputs at least one representative $M''$ with $r(M'') = k'$.
\end{lemma}
\begin{proof}
Let $(G, c)$ be the colored graph with $G = (V, E)$, and let $E_R := R(E)$ be the set of all red edges and $E_B := B(E)$ be the set of all blue edges. The algorithm works as follows: 
\begin{enumerate}
    \item Enumerate all (not necessarily perfect) matchings $X \subseteq E_R$ with $|X \symdif R(M)| \leq s$. For each such $X$, use a classical maximum matching algorithm on the blue edges to check whether there exists $Y \subseteq E_B$ such that $X \dotunion Y =: M''$ is a PM. If the answer is affirmative, we add the number $k' := |X| = r(M'')$ to the output set (together with its representative $M''$).
    \item After that, we repeat the same procedure with the colors switched: Enumerate all  matchings $X \subseteq E_B$ with $|X \symdif B(M)| \leq s$. For each such $X$, check whether there exists $Y \subseteq E_R$ such that $X \dotunion Y$ is a PM. If yes, we add the number $k' := n/2 - |X|$ to the output set (together with its representative $M''$).
\end{enumerate}
The enumeration of sets $X$ can be done in $O(m^s)$ time. Note that this algorithm is sound, in the sense that every PM $M''$ generated by it is indeed contained in $\Nb_s(M)$. On the other hand, the algorithm is also complete: If $M' \in \Nb_s(M)$, then either $R(M')$ or $B(M')$ appears in the enumeration. This means that not necessarily $M'$, but at least some $M''$ with $r(M') = r(M'')$ is found by the algorithm.
The total runtime of the algorithm is $\Theta(m^{s}f_M)$, where $f_M$ denotes the time it takes to solve the perfect matching problem deterministically. For simplification, we let $f_M = \bigO(mn^2) = \bigO(m^3)$ \cite{edmonds1965paths}.
\end{proof}

\SetKwInput{KwInput}{Input}
\begin{algorithm}
\KwInput{Colored graph $(G, c)$, integer $k \geq 0$, local search parameter $s \geq 0$}
\KwResult{Either a PM $M$ with $r(M) = k$, or the info that local search was unsuccessful.}
$\Mmin \gets $PM in $G$ with minimum number of red edges among all PMs \label{algline:M_min}\;

$M \gets \Mmin$\;
\While{$r(M) \neq k$}{
		Try to find $M' \in \Nb_s(M)$ s.t. $r(M) < r(M') \leq k$ using \cref{obs:guess-difference}\;
		\eIf{successful}{
			$M \gets M'$\;
		}{
			return "local search failed."\;
		}
		
	}
	return $M$\;

\caption{A simple local search algorithm, $\local(s)$.}
\label{alg:local-search}
\end{algorithm}
With \cref{obs:guess-difference} in mind, we introduce Algorithm~\ref{alg:local-search} as the most natural local search algorithm. It starts with a PM with the minimum number of red edges and iteratively tries to increase $r(M)$. Note that the PM $\Mmin$ in the first line of the algorithm can be computed in polynomial time (one can run a classical maximum weight perfect matching algorithm, where red edges receive weight -1, and blue edges receive weight 0).  
Algorithm~\ref{alg:local-search} can return false negatives, in the sense that given a yes-instance of {\EM} it is possible for the algorithm to get stuck in a local optimum and return "false". Algorithm~\ref{alg:local-search} can not return false positives.
If we increase the search parameter $s$, we expect Algorithm~\ref{alg:local-search} to be correct more often on average, but we also expect a longer runtime. We denote Algorithm~\ref{alg:local-search} with parameter $s$ by the name $\local(s)$. Since every successful iteration increases $r(M)$, the running time of $\local(s)$ is bounded by $\bigO(m^{s+4})$. It is desirable to understand when $\local(s)$ correctly solves $\EM$. This is partially answered in the next subsection.

\subsection{A Sufficient Condition for Local Search} 

We present a sufficient condition for $\local(s)$ to correctly solve $\EM$. Although the algorithm $\local(s)$ is quite simple, the proof that our condition suffices for correctness of the algorithm is involved. The main idea is the observation that in certain dense and highly structured graphs, it is possible to prove strong locality properties for the set of all perfect matchings. In particular, we consider graphs which  have the following technical property:

\begin{definition}
 \label{def:pshort}
Let $t \geq 2$ be an integer. A graph $G$ has the so-called path-shortening property $\pshort(t)$, if for all PMs $M \subseteq G$, and for all $M$-alternating paths $P$ the following holds: If $F \subseteq  P \cap M$ is a subset of matching edges of size $|F| = t$ and $F = \fromto{\set{a_1,b_1}}{\set{a_t,b_t}}$, where the vertices $a_1,b_1,a_2,b_2,\dots,a_t,b_t$ appear in this order along the path, then the graph $G$ contains an edge $\set{a_i, b_j}$ for some indices $1 \leq i < j \leq t$ or both the edges $\set{a_{i_1},a_{i_3}}, \set{b_{i_2},b_{i_4}}$ for some indices $1 \leq i_1 < i_2 < i_3 < i_4 \leq n$.
\end{definition}
\begin{figure}[htpb]
    \centering
    \includegraphics{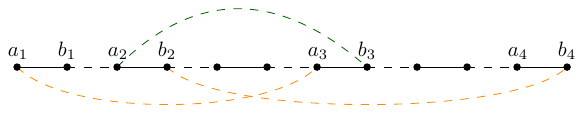}
    \caption{An example of the property $\pshort(4)$ on a path of length 11. Matching edges are bold. Both possibilities of path shortening are highlighted.}
    \label{fig:pshort}
\end{figure}
An illustration of this property is provided in \cref{fig:pshort}. Note that the property is monotone in $t$, i.e.\ $\pshort(t)$ implies $\pshort(t')$ for all $t' > t$.
Our main result is the insight that the property $\pshort$ is sufficient for local search to be correct.

\begin{theorem}
\label{thm:main-thm}
If a graph $G$ has property $\pshort(t)$, then the deterministic algorithm $\local(\bigO(t^{12}))$ solves $\EM$ on graph $G$ in time $n^{\bigO(t^{12})}$ (for all possible edge colorings $c : E(G) \to \set{\text{red},\text{blue}}$ and all target values $k \in \N$).
\end{theorem}

In particular, if $t = O(1)$, then the algorithm above is polynomial-time. 
Such a theorem is only useful of course, if we can show that many different graphs have this mysterious property $\pshort$. Indeed, we can show (proofs are postponed to \cref{subsec:graphs_that_have_pshort}):
\begin{itemize}
\item Complete $r$-partite graphs have the property $\pshort(3)$ for all integers $r \geq 1$. 
\item Graphs of bounded neighborhood diversity $d$ have the property $\pshort(d+1)$.

\item Graphs of bounded independence number $\alpha$ have the property $\pshort(2\Ram(\alpha+1))$, where $\Ram(x) \leq 4^{x}$ is the diagonal Ramsey number.
\item Graphs of bounded bipartite independence number $\beta$ have property $\pshort(2\beta + 2)$.
\item If a graph $G$ has no complete bipartite $t$-hole (i.e.\ the complement $\overline{G}$ does not contain $K_{t,t}$ as subgraph), then $G$ has the property $\pshort(2t)$.
\item The random graph $G(n,p)$ has property $\pshort(2\log(n)/p)$ with high probability.
\end{itemize} 

The proof of \cref{thm:main-thm} is quite technical and requires many steps. The complete proof is given in \cref{sec:proof-main-thm}. The main insight is the observation that in graphs with property $\pshort$, the set of all perfect matchings must obey strong locality guarantees (\cref{subsec:locality_lemma}). For the proof of this locality statement, we introduce the new idea of \emph{local modifiers} (\cref{subsec:modifier}). Each local modifier has a weight associated to it, and the goal becomes to combine the weights in such a way, that they cancel out to be 0. The proof uses ideas and tools from Combinatorics, like an argument similar to the Erd\H{o}s-Szekeres theorem (\cref{subsec:pshort-implies-lmod}), and a helpful lemma from number theory about 0-sum subsequences (compare \cref{subsec:lmod_implies_locality_lemma}). 

\subsection{Limitation of Local Search}
A natural question is whether the approach we presented in this section extends to more graph classes, in particular to all dense graph classes. Here we show that our local search approach fails even for some very dense and very structured classes of graphs. We consider the case of chain graphs. 

Recall the definition of a chain graph from \cref{subsec:graph_class_definitions}. Note that chain graphs can be sparse (e.g. an empty graph is a chain graph), but when required to contain a perfect matching, the graph has to be dense. This can be seen by considering the vertex of highest label on one side and observing that it must be connected to all vertices on the other side. By recursively applying this observation, we can see that the number of edges in the graph must be at least $n^2/8$.

The following is an example of a chain graph that does not satisfy the property $\pshort(t)$ even for $t = n$.
Let $G = (X \dotunion Y, E)$ be the chain graph defined by $|X| = |Y| = n$, $N(x_i):=\{y_{n-i+1},\dots,y_{n}\}$ for all $1 \leq i \leq |X| $ (see \Cref{fig:chainnoshort}). Observe that it is a valid chain graph and that $M := \{\{x_i,y_{n-i+1}\}, \text{ for } 1 \leq i \leq |X|\}$ is a perfect matching. Also observe that there is no edge of the form $\{x_i,y_j\}$ for $1\leq i \leq n - j \leq n$ as required for the property $\pshort(t)$ (since the graph is bipartite, no edges of the form $\{x_i,x_j\}$ or $\{y_i,y_j\}$ exist either).

\begin{figure}[htpb]
    \centering
    \includegraphics{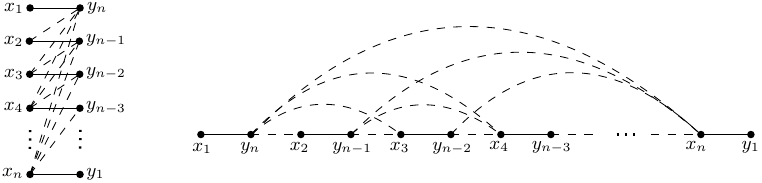}
    \caption{An example of a chain graph $G$ and a perfect matching $M$ in $G$ (left figure), where there exists an $M$-alternating path with $n$ edges from $M$ (right figure). Edges in $M$ are bold. The property $\pshort(n)$ is violated on this path.}
    \label{fig:chainnoshort}
\end{figure}

\section{Extending Karzanov's Characterization}
\label{sec:karzanov}

In this section, we extend the characterization of exact matchings given by Karzanov~\cite{karzanov1987maximum} for complete and complete bipartite graphs to chain graphs, unit interval graphs, and complete $r$-partite graphs. Moreover, we exploit this to give deterministic polynomial-time algorithms for \EM\ on those graph classes. This complements the results of Section~\ref{sec:local_search}, as chain graphs and unit interval graphs are not captured by the local search approach. On the other hand, complete $r$-partite graphs actually fit both frameworks. Note that we only provide a coarse outline here and defer many of the proofs and details to Appendix~\ref{section:proofs_for_karzanov}.

Given a colored graph $(G, c)$, we denote by $k_{\min}(G)$ the smallest integer $k$ such that there is a PM $M$ in $G$ with $r(M) = k$, and by $k_{\max}(G)$ the largest integer $k$ such that there is a PM $M$ in $G$ with $r(M) = k$. Assuming that $G$ admits at least one PM, both $k_{\min}(G)$ and $k_{\max}(G)$ exist.

Karzanov~\cite{karzanov1987maximum} proved that unless a given colored complete or balanced complete bipartite graph $(G, c)$ has a very specific structure, there must be a PM $M$ in $G$ with $r(M) = k$ for all $k_{\min}(G) \leq k \leq k_{\max}(G)$. 
Moreover, he also characterized the special cases where this is violated. In particular, even in those special cases 
the following property still holds. (We use the symbol $\equiv_2$ to denote equivalence modulo 2).

\begin{definition}[Karzanov's Property]
\label{def:karzanov_property}
        A colored graph $(G, c)$ satisfies 
        Karzanov's property if for any two PMs
        $M$ and $M'$ with $r(M) \equiv_2 r(M')$ and any integer $k$ with $r(M) \leq k \leq r(M')$ and $k \equiv_2 r(M) \equiv_2 r(M')$, $G$
        admits a PM $M''$ with $r(M'') = k$.
\end{definition}
In other words, if a graph has Karzanov's property, then if we go in steps of two we always find all possible values of red edges between two given PMs of the same parity. We extend this line of work by proving that all colored chain graphs, unit interval graphs, and complete $r$-partite graphs satisfy Karzanov's Property too.  
This allows us to decide \EM\ on these graph classes by using an algorithm for Bounded Correct-Parity Perfect Matching (\BCPM) introduced by \cite{Maalouly_Steiner_Wulf}.

\begin{quote}
\textbf{Problem} $\BCPM$\\
\textbf{Input:} Colored graph $(G, c)$, integer $k \geq 0$.\\
\textbf{Question:} Is there a PM $M$ in $G$ with $r(M) \leq k$ and $r(M) \equiv_2 k$? 
\end{quote}
We claim that if a graph has Karzanov's property, then {\EM} reduces to \BCPM. 
Indeed, let $\mathcal{A}$ be an algorithm for \BCPM. Given a colored graph $(G, c)$ and an integer $k$, we can call $\mathcal{A}$ with input $(G, c)$ and $k$ to check whether there exists a PM $M$ with $r(M) \leq k$ and $r(M) \equiv_2 k$. Next, assume we compute the inverse coloring $\overline{c}$ with $\overline{c}(e) = \text{red}$ if $c(e) = \text{blue}$, and $\overline{c}(e) = \text{blue}$ if $c(e) = \text{red}$ for all $e \in E$. By calling $\mathcal{A}$ with input $(G, \overline{c})$ and $k' = \frac{n}{2} - k$, we can check the existence of a PM $M$ with $r(M) \geq k$ and $r(M) \equiv_2 k$ in $(G, c)$. If we know that Karzanov's property holds in $(G, c)$, these two pieces of information are sufficient to decide \EM. (Note that instead of using the algorithm for \BCPM\ twice, we can first use an algorithm for the simpler problem \CPM\ which is defined similarly to \BCPM \ but without the bound on the number of red edges. Depending on the number of red edges in the 
output matching we then set the \BCPM \ input appropriately. \CPM\ has been shown to be solvable in deterministic polynomial time on general graphs \cite{Maalouly_Steiner_Wulf}.) 

\begin{observation}
\label{obs:karzanov+bcpm}
    \EM\ reduces to \BCPM\ in colored graphs that satisfy Karzanov's property.
\end{observation} 
Recall that our goal is to give deterministic polynomial-time algorithms for \EM\ in unit interval graphs, chain graphs,
and complete $r$-partite graphs. Observation~\ref{obs:karzanov+bcpm} provides a possible strategy to achieve this. In particular, we will now proceed to give a condition on graphs that implies Karzanov's property. As it turns out, the same condition is also sufficient to give deterministic polynomial-time algorithms for \BCPM.

\subsection{A Sufficient Condition for Karzanov's Property}

We will now present a sufficient condition for Karzanov's property. As it turns out, the condition is also sufficient to give deterministic polynomial-time algorithms for \BCPM. Our condition is based on the existence of certain chord structures in all even-length cycles of a given graph. 
To state it, we first need to introduce some terminology for chords. 

\begin{definition}[Odd Chord, Even Chord, Split of a Chord]
    Let $C$ be a cycle in a graph $G = (V, E)$. An edge $e \in E$ is a chord of $C$ if and only if both endpoints of $e$ are on $C$ but $e \notin C$.
    Let now $e = \{u, v\}$ be a chord of $C$ and consider the paths $P_1, P_2$ obtained by splitting $C$ at $u$ and $v$ i.e.\ $C = P_1 \dotunion P_2$. We call $e$ an odd chord of $C$ if and only if either $P_1$ or $P_2$ has odd length. Otherwise, $e$ is called an even chord of $C$. The split of $e$ is the minimum of the lengths of $P_1$ and $P_2$.
\end{definition}
Note that the above definition technically allows $C$ to have even or odd length, but in general we will only be interested in chords of even-length cycles here.

\begin{definition}[Adjacent Chords]
    Let $C$ be a cycle in a graph $G$ with chords $e = \{x, y\} \in E$ 
    and $f = \{u, v\} \in E$ whose endpoints appear on $C$ in the 
    order $u, v, x, y$. Then $e$ and $f$ are said to be adjacent chords of $C$ if additionally, we have 
    $\{v, x\} \in C$ and $\{u, y\} \in C$.
\end{definition}
\begin{figure}[htpb]
    \centering
    \includegraphics[scale=0.85]{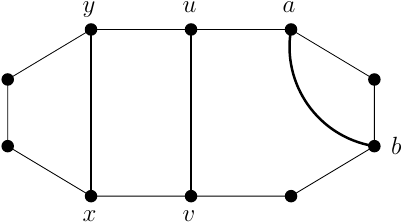}
    \caption{ An example of a $10$-cycle with three chords. The chords $\{x, y\}$ and $\{u, v\}$ are adjacent and they are both odd chords. Conversely, $\{a, b\}$ is an even chord with split $2$. The split of $\{x, y\}$ is $3$ and the split of 
    $\{u, v\}$ is $5$. }
    \label{fig:chords}
\end{figure}
An example of these definitions is found in \cref{fig:chords}. Given these definitions, we are now ready to state our sufficient condition. Note that the condition is only about the graph structure, i.e.\ the coloring is irrelevant here. 

\begin{definition}[Chord Property]
\label{definition:chord-property}
    A simple graph $G = (V, E)$ satisfies the chord property if 
    \begin{enumerate}
        \item every even cycle $C$ of length at least $6$ either has an odd chord or all possible even chords, and
        \item every even cycle $C$ of length at least $8$ either has two adjacent odd chords or all possible even chords with split at least $4$. 
    \end{enumerate}
\end{definition}
As it turns out (see Subsections~\ref{subsection:chain_graph}--\ref{subsection:complete_r-partite}), chain graphs, unit interval graphs, and complete $r$-partite graphs all satisfy the chord property. In fact, the even chords are only needed in complete $r$-partite graphs. In other words, chain graphs and unit interval graphs satisfy the chord property without making use of the parts about even chords. 
By our next lemma, this means that all three graph classes satisfy Karzanov's property for every possible coloring. 

\begin{restatable}[Chord Property is Sufficient]{lemma}{chordproperty}
\label{lemma:chord_property_is_sufficient}
    Let $G$ be an arbitrary graph satisfying the chord property and let $c$ be an arbitrary coloring of $G$. Then the colored graph $(G, c)$ satisfies Karzanov's property. 
\end{restatable}
\begin{proof}
    Deferred to Subsection~\ref{subsection:proof_lemma_chord_property}.
\end{proof}
Note that the reverse is not true, i.e.\ given a colored graph $(G, c)$ with Karzanov's property, it is not necessarily the case that $G$ has the chord property.

In order to solve \EM\ on graphs satisfying the chord property, it remains to give a deterministic polynomial-time algorithm for \BCPM\ on those graphs.

\begin{restatable}[\BCPM\ on Graphs with the Chord Property]{lemma}{polytimebcpm}
\label{lemma:polynomial-time_algorithm_bcpm}
    There is a deterministic polynomial-time algorithm that decides \BCPM\ correctly for all inputs where the graph satisfies the chord property.
\end{restatable}
\begin{proof}
    Deferred to Subsection~\ref{subsection:proof_lemma_bcpm_algo}.
\end{proof}
Finally, we can combine Observation~\ref{obs:karzanov+bcpm} with  Lemma~\ref{lemma:chord_property_is_sufficient} and 
Lemma~\ref{lemma:polynomial-time_algorithm_bcpm} to get the main 
result of this section.
\begin{theorem}
    There is a deterministic polynomial-time algorithm that decides \EM\ on all colored graphs $(G, c)$ where $G$ satisfies the chord property.
\end{theorem}
\begin{proof}
    The colored graph $(G, c)$ satisfies Karzanov's property by Lemma~\ref{lemma:chord_property_is_sufficient}. By Observation~\ref{obs:karzanov+bcpm}, this reduces deciding \EM\ for $(G, c)$ to deciding \BCPM. Moreover, the graph $G$ remains unaltered in this reduction. Hence, this can be achieved in deterministic polynomial-time with the algorithm from Lemma~\ref{lemma:polynomial-time_algorithm_bcpm}.  
\end{proof}
We prove in Subsections~\ref{subsection:chain_graph}--\ref{subsection:complete_r-partite} that unit interval graphs, chain graphs, and complete $r$-partite graphs satisfy the chord property. 
We conclude that \EM\ restricted to those graph classes can be decided in deterministic polynomial-time.

\subsection{Limitation of Karzanov's Property}
A natural question is whether the approach we presented in this section extends to more graph classes. In particular, interval graphs and bipartite interval graphs would be suitable candidates as they are superclasses of unit interval graphs and chain graphs, respectively. 

Unfortunately, it turns out that Karzanov's Property is violated on both graph classes. Concrete counterexamples are given in Figures~\ref{fig:counterexample_interval} and~\ref{fig:counterexample_bip_interval}.
\begin{figure}
    \centering
    \includegraphics[scale=0.85]{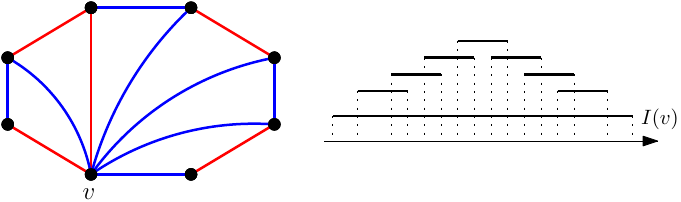}
    \caption{On the left we have a colored interval graph on eight vertices which does not satisfy Karzanov's property. In particular, there are PMs with $0$, $1$, $3$, and $4$ red edges but there is no PM with exactly $2$ red edges. The interval representation of the graph is given on the right. Interval $I(v)$ corresponds to vertex $v$ from the left. Note that the vertical position of the intervals is irrelevant, only the relative horizontal position of the intervals matters.}
    \label{fig:counterexample_interval}
\end{figure}
\begin{figure}
    \centering
    \includegraphics[scale=0.85]{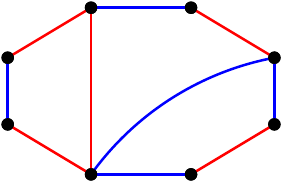}
    \caption{By deleting the even chords from the graph in Figure~\ref{fig:counterexample_interval}, we obtain a bipartite interval graph. It admits the same PMs as the interval graph in Figure~\ref{fig:counterexample_interval}. Hence, this colored graph also violates Karzanov's property. }
    \label{fig:counterexample_bip_interval}
\end{figure}

While our approach fails to generalize to these graph classes, there still seems to be some hope. Consider the following property of colored graphs, which we coined \emph{Karzanov's weak property}.

\begin{definition}[Karzanov's Weak Property]
    A colored graph $(G, c)$ satisfies Karzanov's weak property if for any two PMs 
    $M$ and $M'$ and integer $k$ with $r(M) \leq k \leq r(M')$, there is a PM $M''$ with $r(M'') \in \{k, k + 1\}$.
\end{definition}
The main difference to Karzanov's property is that we are missing the constraint on the parity of the number of red edges. Consider e.g.\ a graph with exactly two PMs with $0$ and $3$ red edges, respectively. Such a graph would satisfy Karzanov's property but violate Karzanov's weak property. In particular, Karzanov's property does not imply Karzanov's weak property. Still, compared to Karzanov's property, Karzanov's weak property gives us less structure to work with and typically holds on larger graph classes. Unfortunately, we have not yet been able to solve \EM\ using Karzanov's weak property.

As it turns out, all colored bipartite chordal and strongly chordal graphs
satisfy Karzanov's weak property.  

\begin{restatable}{lemma}{weakkarzanov}
\label{lemma:weak_karzanov}
    All colored bipartite chordal and strongly chordal graphs satisfy Karzanov's weak property.
\end{restatable}
\begin{proof}
    Deferred to Subsection~\ref{subsec:weak_karzanov}.
\end{proof}

In particular, the same observation holds in all colored interval and bipartite interval graphs as they are subclasses of strongly chordal and bipartite chordal graphs, respectively.

\section{Conclusion}

In this paper we made substantial progress towards solving the notoriously difficult Exact Matching problem, in particular in the regime of dense graphs. 
We provide general frameworks that not only encompass all previously known results for these types of graphs, but also include a multitude of graph classes for which the problem is now solved. We remark that it is inherent to our techniques that they will fail on sparse graphs: 
It seems very unlikely that a local search on a sparse graph is successful (since changing one edge of a PM in a sparse graph often times requires changing many more edges). 
It is also unlikely that a sparse graph has Karzanov's property: On sparse graphs, we do not expect the solution landscape to be dense.
Still, we hope that our approach sheds further light onto these questions. One could imagine, for example, to split a graph into a dense and a sparse part, and apply different techniques to different parts.

In this paper, we also provided some open questions that are reasonable to attack next, since they seem to be in reach of current methods. In particular, is it possible to have a deterministic poly-time algorithm for $G(n,1/2)$? 
Can one find a deterministic poly-time algorithm for those graph classes where the weak Karzanov property holds? (Interval graphs, bipartite interval graphs, strongly chordal graphs, bipartite chordal graphs.) 
Note that for graphs with the weak Karzanov property, we can always find a PM with either $k-1$ or $k$ red edges, but the final decision if $k$ can be achieved still seems difficult. 


\bibliography{references}

\appendix

\section{Path-Shortening implies Successful Local Search}
\label{sec:proof-main-thm}

This section is devoted to the proof of our main theorem (\cref{thm:main-thm}), i.e.\ we prove that the property $\pshort(t)$ implies that local search $\local(O(t^{12}))$ always succeeds. In order to structure the proof, we split it into several parts. First, we introduce a central lemma, which we call the \emph{locality lemma}. If we assume that the locality lemma is true, we can easily prove the main theorem.

The proof of the locality lemma itself is again split into two parts. First, we introduce an intermediate property, called the \emph{local modifier property} $\lmod$ in \cref{subsec:modifier}. We first prove the implication $\pshort \Rightarrow \lmod$. After that we prove in a second step that $\lmod$ implies the locality lemma. We start by giving a few preparational lemmas in \cref{subsec:preparation}. Then we present the locality lemma in \cref{subsec:locality_lemma}. Afterwards, we introduce the intermediate property $\lmod$ in \cref{subsec:modifier}. The proof of the locality lemma is presented in \cref{subsec:pshort-implies-lmod,subsec:lmod_implies_locality_lemma}. Finally, we show in \cref{subsec:graphs_that_have_pshort} that several graph classes have the property $\pshort$.

We remark that we do not optimize for constants in this section, in order to keep the complicated proof as simple as possible. In the end, we obtain the bound $10^52^{12}t^{12} = O(t^{12})$. With a more careful calculation, this bound can likely be improved.
 
\subsection{Some Preparation}
\label{subsec:preparation}
We need some preparation for the following subsections. We start by introducing a crucial tool, the \emph{$M$-induced weight function}. 
\begin{definition}
Let $(G,c)$ be a colored graph and $M \subseteq G$ be a PM. The $M$-induced weight function $w^{(c)}_M : E(G) \to \Z$, is given by
\[ w^{(c)}_M(e) := \begin{cases} -1 & e \text{ is red and } e \in M\\
+1 & e \text{ is red and } e \not\in M\\
0 & e \text{ is blue}.
\end{cases}\]
If the coloring $c$ is clear from context, we denote $w^{(c)}_M$ simply by $w_M$.
\end{definition} 
We claim that the function $w_M$ counts the number of red edges that would change if $M$ is modified, in the following way: If $e$ is red and $e \not\in M$, then adding $e$ increases $r(M)$, hence $w_M(e) = +1$. If $e$ is red and $e \in M$, then removing $e$ decreases $r(M)$, hence $w_M(e) = -1$. If $e$ is blue, $r(M)$ does not change, hence $w_M(e) = 0$. These facts together imply the following observation. For a set $F \subseteq E$ of edges, we define $w_M(F) := \sum_{e \in F}w_M(e)$.
\begin{observation}
\label{obs:wM-counts-difference}
If $M,M'$ are two PMs and $D := M \symdif M'$, then $r(M') = r(M) + w_M(D)$.
\end{observation}
By symmetry, we also have $r(M) = r(M') + w_{M'}(D)$. 

The following lemma concerns what happens if we switch around the colors. Define the inverse colors $\overline{c}(e) = \text{red}$ if $c(e) = \text{blue}$, and $\overline{c}(e) = \text{blue}$ if $c(e) = \text{red}$.
\begin{lemma}
\label{lemma:color-switch}
Let $(G,c)$ be a colored graph, let $M$ be a PM and $P$ be an $M$-alternating path such that $|E(P)|$ is even. Then $w^{(\overline c)}_M(P) = - w^{(c)}_M(P)$.
\end{lemma}
\begin{proof}
Divide the edges of $P$ into $|E(P)|/2$ pairs of a matching edge and a non-matching edge each. For each pair $X := \set{e,e'}$, we do a case distinction. If both $c(e) = c(e') = \text{red}$, then $w^{(c)}(X) = 1 - 1 = 0 + 0 = -w^{(\overline c)}(X)$. Analogous reasoning holds if $c(e) = c(e') = \text{blue}$. If one of the edges $e,e'$ is red and the other one blue, we also have $w^{(c)}(X) = -w^{(\overline c)}(X)$ since exactly one of $e,e'$ is contained in $M$.
\end{proof}

(Note that \cref{lemma:color-switch} is false in general if $|E(P)|$ is odd.) The following observation follows directly from the definitions:
\begin{observation}
Let $(G,c)$ be a colored graph with a PM $M$. Any subset $F \subseteq E(G)$ of edges satisfies the inequalities
\begin{equation}
|w_M(F)| \leq r(F) \leq |F|.
\end{equation}
\end{observation}

\subsection{The Locality Lemma}
\label{subsec:locality_lemma}

In this subsection, we introduce the \emph{locality lemma}. Roughly speaking, the locality lemma states that in any graph with the property $\pshort$, the set of all perfect matchings has the following locality properties. First of all, there can not be too large gaps in the landscape of possible values $r(M)$ (by a gap we mean values $k_1 \ll k_2$ such that there are PMs with $k_1$ and $k_2$ red edges, but all the values between $k_1,k_2$ are impossible to achieve). Secondly, if two PMs $M_1, M_2$ have a similar amount of red edges, i.e.\ if $|r(M_1) - r(M_2)|$ is small, then it is possible to find a representative $M_2'$ of $M_2$ in the local neighborhood of $M_1$.

\begin{restatable}[Locality lemma]{lemma}{restateLocalityLemma}
\label{lem:locality}
    If a graph $G$ has the property $\pshort(t)$ for some $t \geq 2$, then for all red-blue colorings $c : E(G) \to \set{\text{red},\text{blue}}$ and for all PMs $M_1, M_2$, it holds that
    \begin{enumerate}
        \item If $|r(M_1) - r(M_2)| > 8t^2$, then there is a PM $\tilde M$ with $\min\set{r(M_1), r(M_2)} < r(\tilde M) < \max\set{r(M_1),r(M_2)}$.
        
        \item If on the other hand $|r(M_1) - r(M_2)| \leq 8t^2$, then there is a PM $M_2'$ with $r(M_2') = r(M_2)$ and additionally $\dist(M_1, M_2') \leq O(t^{12})$.
    \end{enumerate}
\end{restatable}

The proof of the locality lemma is postponed until the next subsection. Once we have proven the locality lemma, the main theorem follows very quickly, as the following consideration shows.
\begin{proof}[Proof of \cref{thm:main-thm} assuming \cref{lem:locality}]
    Consider a run of the algorithm $\local(O(t^{12}))$ on an input graph with the property $\pshort(t)$. If the instance of {\EM} is a no-instance, then the local-search algorithm will fail to find a solution and terminate after polynomially many rounds. If on the other hand it is a yes-instance, then the algorithm can not get stuck in a local optima: Indeed, as long as the current candidate matching $M$ has the property $r(M) < k$, by the locality lemma (item 1), there exists a PM $M'$ with $r(M) < r(M') \leq k$ and $|r(M) - r(M')| \leq 8t^2$. By the locality lemma (item 2), this PM or its representative is considered by the local search procedure. Hence the local search procedure makes true progress every time until a PM with $r(M) = k$ is found.
\end{proof}

\subsection{Local Modifier Property}
\label{subsec:modifier}
In this section, we introduce the property $\lmod$, which is used as an intermediate step in the proof of the locality lemma. Roughly speaking, the intermediate property $\lmod$ for a graph $G$ states that every modification of a perfect matching which modifies a lot of edges can be equivalently turned into a more local modification, while maintaining certain monotonicity guarantees.

We introduce the notion of a \emph{modifier} of a PM $M$ with respect to some $M$-alternating path $P$.

\begin{definition}
\label{def:modifier}
Let $(G, c)$ be a colored graph, $M$ be a PM, and $P$ be an $M$-alternating path. An $(M, P)$-modifier is an $M$-alternating cycle $C$, such that 
\begin{itemize}
    \item $C$ uses not more vertices than $P$, i.e.\ $V(C) \subseteq V(P)$
    \item $C$ uses at least one non-matching edge from $P$, i.e.\ $C \cap (P \setminus M) \neq \emptyset$.
\end{itemize}
\end{definition}

The main idea behind this definition is that we want to talk about changes to $M$ which are \emph{local}, in the sense that they only change the matching at the vertex set $V(P)$. This is ensured by the first condition. The second condition ensures that every such modification makes some sort of "progress". This is explained in more detail in the next paragraph. 

Let $M$ be a PM, $P$ be an $M$-alternating path and $C$ be an $(M,P)$-modifier. We say the \emph{weight} of the modifier $C$ is $w_M(C)$. We say that \emph{applying} the modifier $C$ to the PM $M$ is the process of substituting $M$ with the new matching $M_\text{new} := M \symdif C$. Note that $M_\text{new}$ is again a PM, because $C$ is an $M$-alternating cycle by definition. The next lemma discusses the effect of applying a modifier. An example of the lemma is shown in \cref{fig:lmod}.

\begin{figure}
    \centering
    \includegraphics{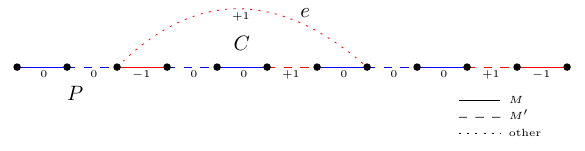}
    \caption{Example for \cref{lemma:apply-modifier}. The path $P$ is an alternating path in $M\symdif M'$. The edge $e$ together with the corresponding subpath creates an $(M, P)$-modifier $C$. The weight function $w_M$ is indicated below the edges. If we define $M_\text{new} := M \symdif C$, then $r(M_\text{new}) = r(M) + w_M(C) = r(M) + 1$.}
    \label{fig:lmod}
\end{figure}
\begin{lemma}
    \label{lemma:apply-modifier}
    Let $(G, c)$ be a colored graph, $M, M'$ be PMs, $P \subseteq M \symdif M'$ be an $M$-alternating path contained in their symmetric difference and $C$ be an $(M,P)$-modifier. Let $M_\text{new} := M \symdif C$ be the PM resulting from applying the modifier $C$ to $M$. Then it holds that
    \begin{enumerate}
        \item $r(M_\text{new}) = r(M) + w_M(C)$
        \item $|M_\text{new} \symdif M'| < |M \symdif M'|$, i.e.\ by applying $C$, we make the PM more similar to $M'$.
    \end{enumerate}
\end{lemma}
\begin{proof}
    The first point is an immediate consequence of \cref{obs:wM-counts-difference}. For the second point, observe that $C$ modifies $M$ only in the induced subgraph $G[V(P)]$. Before the modification, $M$ and $M'$ had zero edges in common in this subgraph. After the modification, $M_\text{new}$ and $M'$ share at least one edge in this subgraph (due to the second condition in the definition of a modifier).
\end{proof}

We have now all the necessary ingredients to define the property $\lmod$.
\begin{definition}
\label{def:t-good}
    Let $(G,c)$ be a colored graph, $M$ be a PM, $P$ be an $M$-alternating path, and $t \geq 0$ be an integer. The path $P$ is called $t$-good, if $|E(P)|$ is even and $w_M(P) \leq 0$ and $r(P) \geq t$.
\end{definition}
\begin{definition}
\label{def:lmod}
    A graph $G$ has the property $\lmod(t)$, if for all red-blue colorings $(G,c)$, for all PMs $M$, and for all $M$-alternating paths $P$ it holds: If $P$ is $t$-good, then there exists an $(M, P)$-modifier $C$ with $-t \leq w_M(C) \leq 0$.
\end{definition}

It is not so easy to give an intuition on why the definition of $t$-good and the property $\lmod$ are chosen to be this specific way. This will hopefully become more clear in the next subsection. For now, let us only say, that they state that local modification of a PM is always possible while at the same time maintaining specific monotonicity guarantees. These turn out to be exactly the guarantees that we want.

The property $\lmod$ is used to extract a negative modifier from a negative path. We would also like to extract a positive modifier from a positive path. To this end, we introduce a property which is dual to  
$\lmod$, called $\dlmod$. It turns out that $\lmod$ and $\dlmod$ are equivalent.

\begin{definition}
\label{def:dual-lmod}
    Let $(G,c), M, P$ be the same as in \cref{def:t-good}. The $M$-alternating path $P$ is called dual-$t$-good, if $|E(P)|$ even and $w_M(P) \geq 0$ and $b(P) \geq t$. The property $\dlmod$ is defined dual to \cref{def:lmod}: If a path $P$ is dual-$t$-good, then there exists an $(M,P)$-modifier with $0 \leq w_M(C) \leq t$. 
\end{definition}

\begin{lemma}
    \label{lem:equivalence-lmod-dlmod}
    $\lmod(t) \Leftrightarrow \dlmod(t)$ for all graphs $G$ and all $t \geq 1$.
\end{lemma}
\begin{proof}
    Assume $G$ has property $\lmod(t)$, we prove $\dlmod(t)$. Let $c$ be a coloring of the edges, $M$ be a PM, and $P$ be a dual-$t$-good $M$-alternating path. We have to prove that there is an $(M,P)$-modifier with $t \geq w^{(c)}_M(C) \geq 0$. Consider the inverse coloring $\overline c$, which switches red and blue colors. Since $|E(P)|$ is even and due to \cref{lemma:color-switch}, with respect to $\overline c$ the path $P$ is $t$-good. The property $\lmod(t)$ is true for all colorings, so in particular it is true for $\overline c$. Hence we find a $(M, P)$-modifier $C$ with respect to $\overline c$ with $-t \leq w^{(\overline c)}_M(C) \leq 0$. By \cref{lemma:color-switch}, we have $t \geq w^{(c)}_M(C) \geq 0$. This proves the implication $\lmod(t) \Rightarrow \dlmod(t)$. The reverse implication is proven analogously.
\end{proof}

\subsection[Pshort implies L-Mod]{$\pshort \Rightarrow \lmod$}
\label{subsec:pshort-implies-lmod}

In this subsection we prove that the path-shortening property implies the local-modifier property. As a helpful tool we consider prefix sums. Consider a fixed colored graph $(G,c)$. If $P$ is an $M$-alternating path with edge set $E(P) = \fromto{e_1}{e_\ell}$, where $e_1,\dots,e_\ell$ appear in this order along the path, we define for all $i = 0,\dots,\ell$ the \emph{$i$-th prefix sum} by \[S_i := \sum_{j=1}^i w_M(e_j).\]
An example is depicted in \cref{fig:prefix-sum}. Note that if $P'$ is a subpath of $P$ starting with edge $e_i$ and ending with edge $e_j$ for $i < j$, then $w_M(P) = S_j - S_{i-1}$.
Consider a subset $F \subseteq M \cap E(P)$ of matching edges along the path, such that $F = \set{e_{i_1},\dots, e_{i_r}}$ with $i_1 < i_2 < \dots < i_r$. The set is called \emph{critical}, if $c(e) = \text{red}$ for all $e \in F$ and $S_{i_1} \geq S_{i_2}  \geq \dots \geq S_{i_r}$. For example, the set $\set{e_4,e_6,e_8}$ in \cref{fig:prefix-sum} is a critical set of size 3. The following lemma shows that $\pshort \Rightarrow \lmod$. Its proof bears some similarities to the proof of the Erd\H{o}s-Szekeres-theorem \cite{erdos1935combinatorial}.    
\begin{figure}[htpb]
    \centering
    \includegraphics{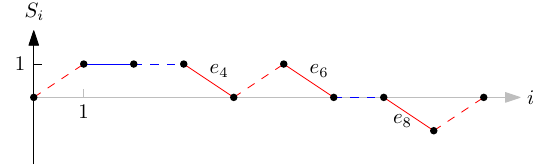}
    \caption{Graphical representation of a prefix sum of an alternating path. The set $\set{e_4,e_6,e_8}$ is a critical set. Matching edges are bold, non-matching edges are dashed.}
    \label{fig:prefix-sum}
\end{figure}

\begin{lemma}
\label{lem:pshort-implies-lmod}
    For every graph $G$ and all $t \geq 1$, the property $\pshort(t)$ implies the property $\lmod(4t^2)$.
\end{lemma}
\begin{proof}
    Assume $G$ has the property $\pshort(t)$. We have to show the property $\lmod(4t^2)$. 
    This means we have to show that for all red-blue colorings $(G,c)$ for all PMs $M$, for all $M$-alternating paths $P$ which are $(4t^2)$-good, there exists a $(M,P)$-modifier $C$ with $-4t^2 \leq w_M(C) \leq 0$. 
    Let $P$ be such a path. Recall that by the definition of $(4t^2)$-good, this means that $|P|$ even, $r(P) \geq 4t^2$ and $w_M(P) \leq 0$. 
    Without loss of generality, if $r(P) > 4t^2$, we shorten the path $P$ until $r(P) = 4t^2$. (If we can find a modifier of the shortened path, it is also a modifier of the complete path). 
    We claim that it in this situation suffices to find a modifier $C$ with $w_M(C) \leq 0$ and this automatically implies $-4t^2 \leq w_M(C) \leq 0$. 
    Indeed, consider the subgraph $G[V(P)]$ induced on the vertex set $V(P)$. By definition of $w_M$, all the negative edges are red edges in $M$. But since $M \subseteq P$ and $r(P) = 4t^2$ and since $C \subseteq V(P)$ by definition of a $(M,P)$-modifier, we have $w_M(C) \geq -4t^2$, as claimed. Therefore, we show in the following only how to obtain a  modifier with $w_M(C) \leq 0$.
    
    We start with the claim that $P \cap M$ contains a critical set of size $t$. In order to show this claim, we do a case distinction.
    
    \textbf{Case 1:} There is some $e_i \in P$ with $S_i \leq -t$.
    In this case, construct a critical set $F$ by traversing the path from start to finish and adding an edge $e_j$ to $F$ whenever $S_j$ attains a new lowest value among all previous prefix sum values. Note that every such edge $e_j$ must be both red and in $M$, because this is the only possibility to contribute -1 to the prefix sum. In total, $F$ is a set of at least $t$ red matching edges with non-increasing corresponding prefix values, hence $F$ is critical.
    
    \textbf{Case 2:} There is some edge $e_i \in P$ with $S_i \geq t$.
    Since in total we have $w_M(P) \leq 0$, and this is equal to the last prefix value $S_\ell$, we can do a similar procedure as in case 1, starting at edge $e_i$.

    \textbf{Case 3:} For all $e_i \in P$ we have $|S_i| \leq t$.
    In this case consider the set $N := R(P) \cap M$ of all red matching edges. Since $w_M(P) \leq 0$, we have $|N| \geq r(P)/2 = 2t^2$. Every edge $e_i \in N$ corresponds to its prefix value $S_i$. By assumption $S_i \in \fromto{-t}{t}$, so by the pigeonhole principle there are at least $|N|/2t \geq t$ edges in $N$ with the same prefix value. The set of all these edges is a critical set of size at least $t$. This concludes the proof of the claim.

    Now that the claim is proven, let $F \subseteq R(P) \cap M$ be a critical set of size $t$. Let $F := \fromto{\set{a_1,b_1}}{\set{a_t,b_t}}$, where the vertices $a_1,b_1,\dots,a_t,b_t$ appear in this order along the path. Now we apply the property $\pshort(t)$ to the set $F$. This leaves us with one of two cases:

    \textbf{Case 1:} There is an edge $\set{a_i,b_j}$ for some indices $i < j$. 
    In this case, let $P' := P[a_i,b_j]$ be the subpath from $a_i$ to $b_j$. Note that $P'$ starts and ends with a matching edge and contains at least one non-matching edge, while the edge $\set{a_i,b_j}$ is not a matching edge. We claim that the cycle $C := P' \cup \set{\set{a_i,b_j}}$ is an $(M,P)$-modifier. Indeed, $C$ is an $M$-alternating cycle with vertex set contained in $V(P)$ and $C \cap P$ contains at least one non-matching edge. Finally, we have $w_M(P') < 0$ by the definition of a critical set. Since the edge $\set{a_i,b_j}$ has a weight of either 0 or +1, we have $w_M(C) = w_M(P') + w_M(\set{a_i,b_j}) \leq 0$.

    \textbf{Case 2:} There are two edges $f_1 := \set{a_{i_1}, a_{i_3}}$ and $f_2 := \set{b_{i_2}, a_{i_4}}$ with $1 \leq i_1 < i_2 < i_3 < i_4$. In this case, the argument is very similar to case 1. We define the subpaths $P'_1 := P[a_{i_1}, b_{i_2}]$ and $P'_2 := P[a_{i_3}, b_{i_4}]$. We note that both subpaths contain a non-matching edge (because $i_1 < i_2$ and $i_3 < i_4$) and $w_M(P'_1) < 0$ and $w_M(P'_2) < 0$. By a similar argument, we see that the cycle consisting out of $P'_1, P'_2$ and the two edges $f_1$ and $f_2$ is an $(M, P)$-modifier (compare \cref{fig:pshort}). Finally, $w_M(C) = w_M(P'_1) + w_M(P'_2) + w_M(f_1) + w_M(f_2) \leq 0$.
    
\end{proof}

\subsection[L-Mod implies Locality Lemma]{$\lmod \Rightarrow$ Locality Lemma}
\label{subsec:lmod_implies_locality_lemma}

In this subsection, we prove the locality lemma. We quickly sketch the main idea. The idea is to show that if $M \symdif M'$ has a large amount of both red and blue edges, then we can find a large amount of $t$-good and dual-$t$-good paths. 
Each of these paths can be transformed into a positive or negative modifier of $M$. In addition to these modifiers, every cycle in $M \symdif M'$ itself can be thought of as a modifier. This idea is sketched in \cref{fig:set-of-modifiers}.
\begin{figure}[htpb]
    \centering
    \includegraphics{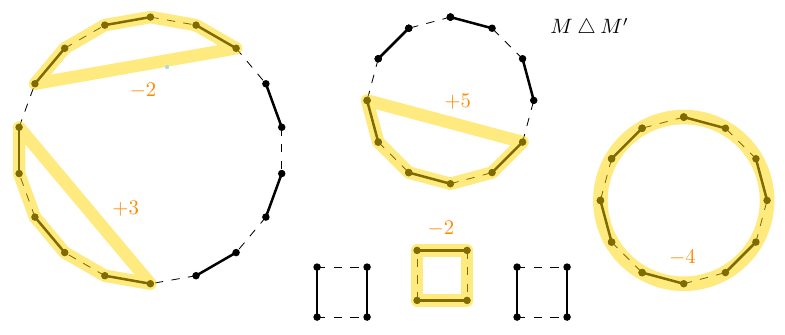}
    \caption{Sketch of the main idea behind the proof of the locality lemma. The idea is to find a set of modifiers, which cancel out to 0. 
    Each modifier is one of the cycles of $M \symdif M'$, or a $(M, P)$-modifier for a subpath $P \subseteq M \symdif M'$. 
    In order to find a set that cancels out to 0, we find many positive and many negative modifiers and apply a number-theoretic result (\cref{lem:zero-sum-subsequence}) to find a 0-sum subset of them.}
    \label{fig:set-of-modifiers}
\end{figure}

Every modifier $C$ has a weight $w_M(C)$ attached to it. Our goal will be to find a set of disjoint modifiers, whose total sum of weight equals 0. Why is this our goal? 
The idea is that for the proof of the locality lemma, we need to show that if for some PMs $M_1, M_2$ they have a large distance $\dist(M_1, M_2) \gg 0$ and $|r(M_1) - r(M_2)|$ is small, then we can find a representative $M_2'$ of $M_2$, such that $r(M_2) = r(M_2')$, but $M_2'$ is contained in the local neighborhood of $M_1$. In order to prove this statement, we apply the above idea to $M := M_2$ and $M' := M_1$. By the above explanation, we find a set $D$ of modifiers with $w_{M_2}(D) = 0$. 
This means we can apply this set $D$ to $M_2$ without changing the number of red edges of $M_2$. However, due to \cref{lemma:apply-modifier}, the quantity $|M_1 \symdif M_2|$ strictly decreases. Hence this process can not go on forever, and at one point we have to find a representative $M_2'$ of $M_2$, which is also in the local neighborhood of $M_1$. 

How do we find a subset of modifiers that cancel out to 0? The idea is to first unconditionally find many positive modifiers, as well as many negative modifiers, and then use a number theoretic result which states that there always exists a subset of them canceling out to 0 (\cref{lem:zero-sum-subsequence}). 
The rest of the subsection is structured as follows: We begin with a helpful lemma, which always guarantees the existence of $t$-good and dual-$t$-good paths (\cref{lem:good-paths-existence}). Then we state the number-theoretic lemma (\cref{lem:zero-sum-subsequence}). 
The main lemma of this section (\cref{lem:weight-0-modifying-set}) explains how to find the 0-sum set of modifiers. Finally, we put all the pieces together and prove the locality lemma. For the sake of reducing the necessary technical details, we do not present the best possible lower and upper bounds for all lemmas in this section. 

We begin with the first lemma. Roughly speaking, it states that a very negative path is $t$-good, a very positive path is dual-$t$-good, and a path who has both a small total weight and a large number of red (blue, respectively) edges contains a $t$-good path (a dual-$t$-good path, respectively). 

\begin{lemma}
    \label{lem:good-paths-existence}
    Let $(G,c)$ be a colored graph, let $M$ be a PM, and $P$ be an $M$-alternating path. Let $t \geq 2, x \geq 0$ be integers.
    \begin{enumerate}
        \item If $w_M(P) \leq -2t$, then $P$ contains a $t$-good path.
        \item If $w_M(P) \geq 2t$, then $P$ contains a dual-$t$-good path.
        \item If $|w_M(P)| \leq x$ and $r(P) \geq 8t^2 +4tx$, then $P$ contains a $t$-good path.
        \item If $|w_M(P)| \leq x$ and $b(P) \geq 8t^2 +4tx$, then $P$ contains a dual-$t$-good path.
    \end{enumerate}
\end{lemma}
\begin{proof}
    For the proof of item 1, we consider the path $P$ itself. If $P$ has an odd number of edges, we delete the last edge of $P$, otherwise we do not change the path $P$. We obtain a path $P'$ such that $|P'|$ is even, and $w_M(P') \leq -2t+1$. Observe that $r(P') \geq 2t - 1 = |w_M(P)| \geq t$, because each unit of weight comes from a red edge. Therefore $P'$ is a $t$-good path.
    
For the proof of item 2, we argue analogously: After possibly deleting the last edge of $P$, we obtain a path $P'$ on an even number of edges with $w_M(P') \geq 2t - 1$. Note that such a path has $b(P') \geq 2t - 2 \geq t$ (any two consecutive red edges together have weight 0, since the path is $M$-alternating. Therefore any unit of +1 weight increase except the last unit must come from a pair of one red and one blue edge.) We conclude that $P'$ is dual-$t$-good.

For the proof of item 3, we assume that the path $P$ consists out of the edges $e_1,e_2,\dots, e_{\ell}$ in this order and we consider the prefix sum $S_i(P) := \sum_{j=1}^iw_M(e_j)$ of $P$ for $i \in \fromto{1}{\ell}$. Let $A := \max\set{ |S_i(P)| : 1 \leq i \leq \ell}$ be the maximum attained absolute prefix sum. We distinguish two cases.

\textbf{Case 1:} $A \geq 2t + x$. In this case, since $|S_i|$ is large for some $i$, but the final absolute value $|S_\ell| = |w_M(P)| \leq x$ is small, we can divide $P$ into two paths $P_1$ and $P_2$ with $w_M(P_1) \geq 2t$ and $w_M(P_2) \leq -2t$. Using item 1.\ on path $P_1$, we are done.

\textbf{Case 2:} $A < 2t + x$. Let $p := 4t+2x$. Due to the bound $r(P) \geq (4t+2x)(2t) = p(2t)$, we can partition the path $P$ into $p$ consecutive subpaths $Q_1,\dots,Q_p$ and possibly one remainder path $\tilde Q$ such that each subpath except the remainder has $2t$ red edges, i.e.\ $r(Q_i) = 2t$ for all $i = 1,\dots,p$. We claim that it is impossible that for all $i \in \fromto{1}{p}$ we have $w_M(Q_i) > 0$. Indeed, if this happens, then we have $p$ consecutive subpaths of the path $P$ of strictly positive weight, but this is a contradiction to the fact that the prefix sum has bounded absolute value $A < 2t+x = p/2$. Therefore we find an index $i \in \fromto{1}{p}$ with $r(Q_i) = 2t$ and $w_M(Q_i)  \leq 0$, hence $Q_i$ is $t$-good. 

Finally, item 4.\ is proven analogously to item 3.
    
\end{proof}

Let $M,M'$ be two PMs and consider their symmetric difference $M \symdif M'$. The previous lemma guarantees that inside $M \symdif M'$ we can find many $t$-good and dual-$t$-good paths $P \subseteq M \symdif M'$ under certain conditions. If the graph has the property $\lmod$, these can be turned into $(M,P)$-modifiers. Our goal is to find a set of modifiers which sums up to 0. In order to achieve this goal, we also consider the cycles (i.e. the connected components) of $M \symdif M'$. These cycles themselves can be viewed as a modifier.

\begin{definition}
Let $(G,c)$ be a colored graph and $M, M'$ be two PMs. An $(M,M')$-modifier is an $M$-alternating cycle $C$ such that either $C$ is one of the cycles in $M \symdif M'$, or there exists some path $P \subseteq M \symdif M'$ such that $C$ is a $(M,P)$-modifier. An $(M,M')$-modifying set is a subset $D \subset E(G)$, such that every of its connected components is an $(M,M')$-modifier. 
\end{definition}

By this definition, an $(M,M')$-modifying set $D$ consists out of vertex-disjoint $M$-alternating cycles. We define $w_M(D) := \sum_{e \in D} w_M(e)$. Analogous to \cref{lemma:apply-modifier}, the following holds true: If $D$ is an $(M, M')$-modifying set, and we define $M_\new := M \symdif D$, then $M_\new$ is a PM with $r(M_\new) = r(M) + w_M(D)$ and $|M_\new \symdif M'| < |M \symdif M'|$. This means that an $(M,M')$-modifying set of weight 0 can be used to reduce the distance between $M$ and $M'$.

\begin{lemma}
\label{lem:zero-sum-subsequence}
    Let $p \in \N$ be a parameter, let $a = (a_1,\dots, a_{2p+2}) \in \fromto{-p}{p}^{2p+2}$ be a sequence of integers in the range $[-p,p]$ of length $2p+2$. If $a$ contains at least $p+1$ nonnegative entries and at least $p+1$ nonpositive entries, then there is a subsequence of $a$ that sums up to 0.
\end{lemma}
\begin{proof}
    If $0 \in a$, we are immediately done. So we assume for the rest of the argument that $0 \not \in a$. 
    We partition the sequence $a$ into the sequence $a^+$, which contains all the positive numbers from $a$ (in some arbitrary order), and the sequence $a^-$, which contains all the negative numbers from $a$ (in some arbitrary order). 
    We construct a new sequence $b$ by merging $a^+$ and $a^-$: For the beginning, we consider $b$ as the empty sequence with prefix sum $S_0 = 0$. After that, we define $b$ recursively: For all $i \geq 1$, the $i$-th entry $b_i$ depends on the prefix sum $S_{i-1} := \sum_{j=1}^{i-1}b_j$. If $S_{i-1} \geq 0$, then we delete the currently first entry of $a^-$ and let $b_i$ be equal to that entry. If $S_{i-1} < 0$, then we delete the currently first entry of $a^+$ and let $b_i$ be equal to that entry. The process ends when either $a^+$ or $a^-$ is empty.
    Note that by construction, we ensure that as long as the process continues, we have $S_i \in [-p, p]$. We now make a case distinction.
    
    \textbf{Case 1:} If $S_i = 0$ for some index $i \neq 0$, then we obtain a subsequence of $b$ (and also of $a$) which sums up to 0, so we are immediately done.

    \textbf{Case 2:} Assume $S_i \neq 0$ for all indices $i$. We consider the process of constructing $b$. The process ends, when either $a^+$ or $a^-$ is empty. Note that if $a^+$ is empty, we must have encountered $p+1$ different indices $i$ with $S_i \in [-p, -1]$, so by pigeonhole principle we have $S_i = S_{i'}$ for different indices $i,i'$. But this means we find a subsequence of $b$ with sum 0, so we are done. If $a^-$ is empty, we must have encountered $p+1$ different indices $i$ with $S_i \in [0, p]$, but since $S_i = 0$ is not possible by assumption, we have encountered $p+1$ different indices $i$ with $S_i \in [1,p]$. So again by the pigeonhole principle we are done.
\end{proof}

\begin{lemma}
\label{lem:weight-0-modifying-set}
Let $G$ be a graph with the property $\lmod(t)$ for some $t \geq 2$, let $(G,c)$ be a red-blue coloring of $G$ and $M,M'$ be two PMs. If $|r(M) - r(M')| \leq 2t$ and $\dist(M,M') \geq 10^5t^6$, then there exists an $(M,M')$-modifying set $D$ of weight $w_M(D) = 0$.
\end{lemma}
\begin{proof}
    We start by giving the general idea behind the proof:
    Let $f(t) \geq 0$ be some function dependent on $t$ specified later. 
    The rough strategy will be to make several case distinctions and find in each case many positive modifiers, each with a weight in the range $[0, f(t)]$, as well as many negative modifiers, each with a weight in the range $[-f(t), 0]$. We make sure that all these modifiers are vertex-disjoint. Since we have many positive, as well as many negative modifiers, we can apply \cref{lem:zero-sum-subsequence} to find a subset of the modifiers, which sums up to 0. 
    In all cases, these modifiers arise either from a $t$-good path, or a dual-$t$-good path, or the modifier is equal to one of the cycles in $M \symdif M'$. We are now ready to start the proof.
    
    We consider the set $M \symdif M'$. Since $\dist(M,M')$ is large, by definition this means there are a lot of both red and blue edges in $M \symdif M'$. Indeed, we have $r(M \symdif M') \geq 10^5t^6$ and $b(M \symdif M') \geq 10^5t^6$. 
    Furthermore, consider the unique decomposition of $M \symdif M' = C_1 \dotunion \dots \dotunion C_\ell$ into vertex-disjoint cycles (where $\ell$ is the number of cycles).
    
    Let $\Cyc := \fromto{C_1}{C_\ell}$ be the set of all cycles, $\Cyc_0 := \set{C \in \Cyc : w_M(C) = 0}$ be the set of all weight-0 cycles, $\Cyc_+ := \set{C \in \Cyc : w_M(C) > 0}$ be the set of all positive cycles, and $\Cyc_- := \set{C \in \Cyc : w_M(C) < 0}$ be the set of all negative cycles of $\Cyc$. Note that if $\Cyc_0 \neq \emptyset$, then any $C \in \Cyc_0$ is a $(M,M')$-modifier of weight 0, and we have proven the lemma. Hence we assume for the following that $\Cyc_0 = \emptyset$. 
    We know that
    \begin{equation}
        |\sum_{i=1}^\ell w_M(C_i)| = |r(M') - r(M)| \leq 2t. \label{eq:lmod-cycle-sum}
    \end{equation}
     We now consider the total weight of all positive cycles and continue with a case distinction on this value.

    \textbf{Case 1:} We have $\sum_{C \in \Cyc_+} w_M(C) \geq 20t^2 + 2t$.  In this case, due to \cref{eq:lmod-cycle-sum}, we know that $\sum_{C \in \Cyc_-}w_M(C) \leq -20t^2$. 
    Let $\ell_+ := |\Cyc_+|$ be the number of positive cycles and let $\Cyc_+ =: \fromto{C^+_1}{C^+_{\ell_+}}$ be the set of all positive cycles. 
    From every positive cycle $C^+_i \in \Cyc_+$, we extract one or more modifiers the following way: If $w_M(C^+_i) \leq 2t$, we consider the cycle $C^+_i$ itself as a $(M,M')$-modifier with weight in the range $[0, 2t]$. 
    In the other case, i.e.\ $w_M(C^+_i) > 2t$, we define $W := w_M(C^+_i)$ and split the cycle $C^+_i$ into $\lfloor \frac{W}{2t} \rfloor$ subpaths, each having weight at least $2t$. Due to \cref{lem:good-paths-existence} (item 1.), every of these subpaths contains a dual-$t$-good path, which due to property $\dlmod(t)$ can be turned into a nonnegative $(M,M')$-modifier. Since the subpaths are disjoint, the modifiers are as well. 
    In total, we extract from the cycle $C^+_i$ at least $\lfloor \frac{W}{2t} \rfloor \geq \frac{W}{4t}$ different modifiers, each with weight in the range $[0, t]$. 
    Combining the two arguments above, and using the fact $\sum_{C \in \Cyc_+}w_M(C) \geq 20t^2$, this means that from the set of cycles $\Cyc_+$, we can extract at least $20t^2/4t = 5t$ different $(M,M')$-modifiers, each with a weight in the range $[0, 2t]$.
    By an analogous argument, we can obtain at least $5t$ different $(M,M')$-modifiers, each with a weight in the range $[-2t, 0]$. Note that all obtained modifiers are disjoint. Using \cref{lem:zero-sum-subsequence} (with parameter $p = 2t$), we find a subset of the modifiers which sums up to 0.

    \textbf{Case 2:} We have $\sum_{C \in \Cyc_-} w_M(C) \leq -20t^2 - 2t$. This argument is analogous to Case 1.

    \textbf{Case 3:} We have $\sum_{C \in \Cyc_+} w_M(C) < 20t^2 + 2t$ and $\sum_{C \in \Cyc_-} w_M(C) > -20t^2 - 2t$. Note that in particular this implies that for every cycle $C \in \Cyc$, its absolute weight is bounded by $|w_M(C)| < 20t^2 + 2t \leq 22t^2$. 
    Furthermore, since we assumed there are no cycles of weight 0, this implies that the number of cycles is bounded by $|\Cyc_+| < 20t^2 + 2t$ and $|\Cyc_-| < 20t^2 + 2t$, hence $|\Cyc| < 40t^2 + 4t \leq 44t^2$. Using the pigeonhole principle, as well as the fact that $\dist(M,M')$ is large, we conclude that there exists a cycle $C_R \in \Cyc$ with many red edges, precisely 
    \[ r(C_R) \geq \frac{r(M \Delta M')}{  |\Cyc| } \geq \frac{10^5t^6}{44t^2} \geq \frac{10^4t^6}{50t^2} = 2000 t^4.\]
    By an analogous argument, there exists a cycle $C_B \in \Cyc$ with many blue edges, $b(C_B) \geq 2000t^4$. Our goal is to extract many negative modifiers from $C_R$, and many positive modifiers from $C_B$. We make a case distinction on whether the two cycles $C_R,C_B$ are the same.

    \textbf{Case 3a:} We have $C_R \neq C_B$. We first show how to extract the negative modifiers from $C_R$. We have $|w_M(C_R)| \leq 44t^2$ and $r(C_R) \geq 2000t^4$ by the arguments above. 
    Let $p := |C_R|$ be the length of cycle $C_R$ and $e_1,\dots,e_p$ be the edges of $C_R$ in this order. 
    We consider for $i \in \fromto{0}{p}$ the prefix sum $S_i := \sum_{j=0}^i w_M(e_i)$. We also consider the maximum attained absolute prefix sum $A := \max \set{|S_i| : i \in \fromto{0}{p}}$. 
    
    \textbf{Case 3aa:} If $A > 50t^2$, then  since the maximum prefix sum is large, but the total weight $|w_M(C_R)| \leq 44t^2$ is small, we find   two subpaths $P_1, P_2$ of $C_R$ with $w_M(P_1) \geq 6t^2$ and $w_M(P_2) \leq - 6t^2$. Using \cref{lem:good-paths-existence} (item 1), we can find at least $3t$ disjoint $t$-good paths and $3t$ dual-$t$-good paths. 
    Due to property $\lmod$, we can turn the dual-$t$-good paths into modifiers with weight in the range $[0,t]$. 

    \textbf{Case 3ab:} If $A \leq 50t^2$, then any subpath $P'$ of the cycle $C_R$ has bounded absolute weight $|w_M(P')| \leq 2A = 100t^2$. Let $x := 100t^2$. 
    We partition $C_R$ into subpaths $\set{P'_i}$, such that each subpath $P'_i$ has $500t^3$ red edges, $r(P'_i) = 500t^3$. We find at least $r(C_R) / (500t^3) \geq 4t$ of these subpaths. Note that $500t^3 \geq 8t^2 + 4tx$, hence we can apply \cref{lem:good-paths-existence} (item 3) to find a total set of $4t$ disjoint dual-$t$-good paths. 

    In both cases 3aa and 3ab we can apply an analogous procedure to $C_B$ to find at least $3t$ disjoint paths in $C_B$ which are $t$-good and turn them into negative modifiers with a weight in the range $[-t,0]$. Since $C_R \neq C_B$, the modifiers are disjoint. Finally, we can apply \cref{lem:zero-sum-subsequence} to find a subset of them with a total sum of 0.

    \textbf{Case 3b:} We have $C_R = C_B$. In this case, we partition the cycle $C_R = C_B$ into two paths $P_1,P_2$, one of which has many red edges, and the other having many blue edges the following way: 
    We start the path $P_1$ with an arbitrary edge, and we continue until the first time when $r(P_1) = 1000t^4$ or $b(P_1) = 1000t^4$ holds. 
    We then let $P_2 := C_R \setminus P_1$. The rest of the proof continues in an analogous fashion. 
    Indeed, we find at least $3t/2$ modifiers with a weight in the range $[0,t]$ and at least $3t/2$ modifiers with a weight in the range $[-t,0]$, so by \cref{lem:zero-sum-subsequence}, we can combine a subset of them to get a total sum of 0. 
\end{proof}

We now have all the necessary ingredients ready to prove the locality lemma, \cref{lem:locality}. We recall the statement.

\restateLocalityLemma*
\begin{proof}
    First, note that property $\pshort(t)$ implies $\lmod(4t^2)$ by \cref{lem:pshort-implies-lmod}. For the proof of the first point, let $M_1, M_2$ in $G$ with $|r(M_1) - r(M_2)| > 8t^2$. Let w.l.o.g. $r(M_1) < r(M_2)$, otherwise we swap the names of $M_1,M_2$. 
    We consider $M_1 \symdif M_2$ and let $C_1 \cup \ldots \cup C_\ell = M_1\symdif M_2$ be its decomposition into alternating cycles. Since $r(M_2) > r(M_1)$, there is at least one index $i$ such that $w_{M_1}(C_i) > 0$. 
    We distinguish two cases: If $w_{M_1}(C_i) < 8t^2$, then $\tilde M := M_1 \symdif C_i$ is a PM with the desired properties. Otherwise we have $w_{M_1}(C_i) \geq 8t^2$. By \cref{lem:good-paths-existence} applied to $C_i$ and the property $\lmod(4t^2)$, we find a $(M_1, M_2)$-modifier $C$ with $w_{M_1}(C) \in [0, 4t^2]$. 
    If $w_{M_1}(C) \neq 0$, the PM $\tilde M := M_1 \symdif C$ has the desired properties. If $w_{M_1}(C) = 0$, we define $M_1' = M_1 \symdif C$. Note that due to \cref{lemma:apply-modifier}, we have $|M_1' \symdif M_2| < |M_1 \symdif M_2|$. 
    Hence this case can only occur a finite amount of times. In the end, the first case must occur and we find a PM $\tilde M$ with the desired properties.

    For the proof of the second point, consider two PMs $M_1, M_2$ with $|r(M_1) - r(M_2)| \leq 8t^2$. We claim that there exists a PM $M_2'$ with the same amount or red edges as $M_2$, and 
    \[ \dist(M_1, M_2') \leq 10^5(4t^2)^6 = 10^52^{12}t^{12} = O(t^{12}). \]
    Indeed, assume that this is not already the case for $M_2$ and currently $\dist(M_1, M_2) > 10^5 2^{12}t^{12}.$ We apply \cref{lem:weight-0-modifying-set} such that the pair $(M,M')$ from the lemma statement is $(M,M') = (M_2, M_1)$ (note that $M_2$ comes before $M_1$). We check that the conditions of the lemma are satisfied: 
    The graph has the property $\lmod(4t^2)$, we have $|r(M_2) - r(M_1)| \leq 2 \cdot 4t^2$, and we have $\dist(M_2, M_1) = \dist(M_1, M_2) > 10^5(4t^2)^6$.
    Since the conditions of the lemma are satisfied, we find a $(M_2,M_1)$-modifying set $D$ of weight $w_{M_2}(D) = 0$. We define $M_2'' :=  M_2 \symdif D$. Because the modifying set has weight 0, we still have $r(M_2'') = r(M_2)$, but due to  \cref{lemma:apply-modifier}, we have $|M_2'' \symdif M_1| < |M_2 \symdif M_1|$. Now we repeat the same procedure starting with PM $M_2''$. Since this process has to stop after a finite amount of steps, at some point we encounter a PM $M_2'$ such that $r(M_2') = r(M_2)$ and $\dist(M_2', M_1) \leq O(t^{12})$. This was to prove.
\end{proof}

\subsection{Showing Property Pshort for Various Graph Classes}
\label{subsec:graphs_that_have_pshort}

Our main theorem, \cref{thm:main-thm}, establishes that for every graph, or every class of graphs with the property $\pshort(t)$, the local search algorithm $\local(O(t^{12}))$ is successful. In this section, we complement this result by showing that various classes of graphs have the property $\pshort$. Since all the proofs in this subsection follow the same pattern, we use the same notation each time. 
Namely, we let $M$ be a PM of the graph $G$ in question, we let $P$ be some $M$-alternating path, and $F \subseteq P \cap M$ be some subset of the matching edges of the path. 
We let $t \geq 2$ be the integer for which we want to prove property $\pshort(t)$, and we assume that $|F| \geq t$, and $F = \fromto{\set{a_1,b_1}}{\set{a_t,b_t}}$, where the vertices $a_1,b_1,a_2,b_2,\dots,a_t,b_t$ appear in this order on the path. 
We have to show that there exist indices, such that either $G$ contains the edge $\set{a_{i_1},b_{i_2}}$ for some $i_1 < i_2$, or  $G$ contains both the edges $\set{a_{i_1}, a_{i_3}}$ and $\set{b_{i_2}, b_{i_4}}$ for some $i_1 < i_2 < i_3 < i_4$.

\begin{lemma}
    A complete $r$-partite graph, for any $r$, has the property $\pshort(3)$.
\end{lemma}
\begin{proof}
    Let $t = 3$. For every vertex $v \in V(G)$, let $p(v) \in \N$ denote the part in which this vertex is located. This means $\set{u,v}$ is an edge if and only if $p(u) \neq p(v)$. We do a case distinction: If $p(a_1) \neq p(b_2)$, we find the edge $\set{a_1, b_2}$ and we are done. A similar argument holds if $p(a_2) \neq p(b_3)$. Hence we assume $p(a_1) = p(b_2)$ and $p(a_2) = p(b_3)$. But clearly $\set{a_2,b_2}$ is an edge, hence $p(a_1) = p(b_2) \neq p(a_2) = p(b_3)$. 
    Hence $\set{a_1, b_3} \in E(G)$ and we are done.
\end{proof}

We recall the definition of neighborhood diversity \cite{lampis2012algorithmic}: In a graph $G$, two vertices $u$ and $v$ are said to have the same type if and only if $N(u) \setminus \set{v} = N(v) \setminus \set{u}$. A graph $G$ has neighborhood diversity $d$ if there exists a partition of $V(G)$ into $d$ sets $T_1, \dots, T_d$, such that all the vertices in each set have the same type. (This alternate definition is easily seen to be equivalent to the one presented in \cref{subsec:graph_class_definitions}.)
\begin{lemma}
    A graph of bounded neighborhood diversity $d$ has the property $\pshort(d+1)$.
\end{lemma}
\begin{proof}
    Let a partition $T_1,\dots,T_d$ with the properties described above be given. For every vertex $v \in V(G)$, let $p(v) \in \fromto{1}{d}$ denote the part in which this vertex is located, i.e.\ $p(v) = i$ such that $v \in T_i$. 
    Since there are $d+1$ edges in $F$, we find two edges $\set{a_i, b_i}, \set{a_j, b_j} \in F$ with $p(a_i) = p(a_j)$, where w.l.o.g.\ $i < j$. Since vertex $a_i$ and $a_j$ have the same type and since $b_j \in N(a_j)$, we also have $b_j \in N(a_i)$. Hence we have found the desired edge $\set{a_i, b_j}$.
\end{proof}

\begin{lemma}
    A graph of bounded independence number $\alpha$ has the property \\ 
    $\pshort(2\Ram(\alpha+1))$, where $\Ram(x) \leq 4^x$ is the diagonal Ramsey number.
\end{lemma}
\begin{proof}
    Let $t = 2\Ram(\alpha + 1)$. We have $F = \fromto{\set{a_1,b_1}}{\set{a_t,b_t}}$. We partition $F$ into $t/2$ pairs of consecutive elements, i.e.\ for $i = 1,\dots,t/2$ define $F_i := \set{\set{a_{2i-1}, b_{2i-1}}, \set{a_{2i}, b_{2i}}}$. 
    Consider the set of vertices $A := \set{a_{2i-1} : i \in \fromto{1}{t/2}}$. 
    This set of vertices has size $\Ram(\alpha + 1)$, so it contains either an independent set or a clique of size $\alpha + 1$. The first is not possible, hence we find a clique $A' \subseteq A$ of size $|A'| = \alpha + 1$. We consider the corresponding vertex set $B' := \set{b_{2i} : a_{2i-1} \in A'}$. 
    Since $|B'| = \alpha + 1$, there is an edge in the induced graph $G[B']$, i.e. an edge $\set{b_{2i}, b_{2i'}}$ for some $i,i'$. 
    Since the corresponding vertices $a_{2i-1}$ and $a_{2i'-1}$ are in the clique $A'$, the edge between them exists, and we have four indices $2i-1 < 2i < 2i'-1 < 2i'$ with properties as requested for $\pshort$. 
\end{proof}

We note that despite the authors best efforts, we were not able to get rid of the exponential bound $\Ram(\alpha +1)$ in the previous proof.

Recall that a bipartite graph on vertex set $A \cup B$ has bounded bipartite independence number, if there is no balanced independent set, i.e.\ sets of vertices $A' \subseteq A, B' \subseteq B$ of size $|A'| = |B'| = \beta+1$ and no edge between $A'$ and $B'$.

\begin{lemma}
\label{lem:bounded-bip-ind}
    A graph of bounded bipartite independence number $\beta$ has the property $\pshort( 2\beta + 2)$.
\end{lemma}
\begin{proof}
    Let $t := 2 \beta + 2$. Consider $F = \fromto{\set{a_1,b_1}}{\set{a_t,b_t}}$. 
    Since the path $P$ is $M$-alternating, because $F \subseteq M \cap P$, and because the graph is bipartite, we observe that all the vertices $a_1,\dots,a_t$ belong to the same bipartition and all the vertices $b_1,\dots,b_t$ belong to the other bipartition. 
    Let $A' := \fromto{a_1}{a_{\beta+1}}$ and $B' := \fromto{b_{\beta+2}}{b_{2\beta+2}}$. 
    Since $|A'| = |B'| = \beta + 1$, and since $(A', B')$ can not be a balanced independent set, there must exist an edge between $A'$ and $B'$. This proves the property $\pshort$.
\end{proof}

\begin{corollary}
\label{cor:t-hole}
    If a graph $G$ has no complete bipartite $t$-hole, i.e. the complete bipartite graph is not part of the complement, $K_{t,t} \not\subseteq \overline G$, then $G$ has property $\pshort(t)$.
\end{corollary} 
\begin{proof}
    The proof is identical to the proof of the previous \cref{lem:bounded-bip-ind}.
\end{proof}

For the next proof, we consider the random graph $G(n,p)$ on $n$ vertices, where every edge appears independently with probability $p \in (0,1)$. We allow for the case that $p = p(n)$ is a function of $n$.

\begin{lemma}
    The random graph $G(n,p)$ has property $\pshort(3 \log(n) / p)$ with high probability.
\end{lemma}
\begin{proof}
    Let $t = 3 \log(n) / p$. We prove that with high probability $G(n,p)$ does not contain $K_{t,t}$ in the complement. By the previous \cref{cor:t-hole}, this is sufficient.
    The proof is a standard union-bound argument: For a fixed choice of the embedding of the vertices of $K_{t,t}$, the probability that $K_{t,t}$ appears in the complement equals $(1-p)^{t^2}$. The possible embeddings of $K_{t,t}$ is at most $\binom{n}{t}\binom{n}{t}$. By a union-bound argument, the probability that there is at least one $K_{t,t} \subseteq \overline G$ is at most
    \begin{align*}
        (1-p)^{t^2}\binom{n}{t}\binom{n}{t}  &\leq e^{-pt^2}n^tn^t \\
        &= e^{- t^2p + 2t\log n} = e^{(-9\log^2 n  + 6 \log^2 n)/p} \\
        &\leq e^{-3\log^2 n} \leq \frac 1 n \text{ \quad for $n \to \infty.$}
    \end{align*}  
    Therefore, with high probability, we have $K_{t,t} \not\subseteq \overline G$ and therefore property $\pshort(t)$ holds with high probability.
\end{proof}

\section{Proofs for Section~\ref{sec:karzanov}}
\label{section:proofs_for_karzanov}

In Section~\ref{sec:karzanov}, we gave an overview over our second framework for obtaining deterministic polynomial-time algorithms for \EM\ restricted to specific graph classes. This framework works on all graphs that have the chord property. In particular, the first insight is that all colored graphs satisfying the chord property also satisfy Karzanov's property (Lemma~\ref{lemma:chord_property_is_sufficient}). This implies that \EM\ on these graphs reduces to the easier problem \BCPM\ (Observation~\ref{obs:karzanov+bcpm}). Hence, the second tool we need is an algorithm for \BCPM\ that works on all colored graphs satisfying the chord property (Lemma~\ref{lemma:polynomial-time_algorithm_bcpm}). 

In this section, we provide the missing proofs for Lemma~\ref{lemma:chord_property_is_sufficient} in Subsection~\ref{subsection:proof_lemma_chord_property}, and for Lemma~\ref{lemma:polynomial-time_algorithm_bcpm} in Subsection~\ref{subsection:proof_lemma_bcpm_algo}. Moreover, we prove in Subsections~\ref{subsection:chain_graph}--\ref{subsection:complete_r-partite} that chain graphs, unit interval graphs, and complete $r$-partite graphs satisfy the chord property. Hence, our framework applies to these graph classes. Finally, we give the missing proof of Lemma~\ref{lemma:weak_karzanov} about Karzanov's weak property in bipartite chordal and strongly chordal graphs in Subsection~\ref{subsec:weak_karzanov}.

\subsection{Chord Property Implies Karzanov's Property}
\label{subsection:proof_lemma_chord_property}

In this subsection, we prove that any colored graph that satisfies the chord property also satisfies Karzanov's property. Our proof will roughly work as follows: We will introduce a concept that we call the rank of a symmetric difference of two PMs. Intuitively speaking, the rank measures in a very specific way how similar the two PMs are. Now consider a colored graph $(G, c)$ with PMs $M, M'$ and integer $k$ such that $r(M) \leq  k \leq r(M')$ and $r(M) \equiv_2 k \equiv_2 r(M')$. Assuming that $G$ satisfies the chord property, we want to prove that there must be a PM $M^*$ with $r(M^*) = k$. For this, we can first assume that $M$ and $M'$ are chosen according to the above constraints such that additionally, $\rank(M \symdif M')$ is minimized. It now remains to prove that we must have either $r(M) = k$ or $r(M') = k$. We achieve this with an indirect proof: Assuming instead that we have $r(M) + 4 \leq r(M')$, we show that we could find a pair of PMs that satisfy all the constraints and have smaller rank. 

The rest of this subsection is organised as follows: We start by introducing \emph{simple modifiers} and \emph{cross modifiers}. These are induced by chords of alternating cycles and can be used to modify PMs. In contrast to the modifier techniques of the local search framework, the modifiers in this section are not necessarily small or bounded. Since our modifiers are induced by chords, the chord property will guarantee us the existence of enough modifiers. Moreover, simple modifiers are also used to define the rank of symmetric differences. This notion of rank is a crucial concept for the whole proof. After introducing these new concepts, we proceed to prove the preparatory Lemma~\ref{lemma:properties_modifiers} which illustrates for the first time how we can modify two given PMs to reduce the rank of their symmetric difference. This is then used to further show in Lemma~\ref{lemma:rank_reduction} that the parity of the PMs can be maintained throughout the process. After that it remains to put the pieces together to get Lemma~\ref{lemma:chord_property_is_sufficient}.

\begin{definition}[Simple Modifier] 
\label{definition:simple_modifier}
    Let $G$ be a graph with a PM $M$ and consider an $M$-alternating cycle $C$. An $M$-alternating cycle $C'$ is a simple modifier of $C$ if and only if $C' \setminus C = \{e\}$ for some (necessarily odd) chord $e \in E$ of $C$.
\end{definition}

\begin{definition}[Cross Modifier]
\label{definition:cross_modifier}
    Let $G$ be a graph with a PM $M$ and consider an $M$-alternating cycle $C$. An $M$-alternating cycle $C'$ is a cross modifier of $C$ if and only if $C' \setminus C = \{e, f\}$ for two even chords $e, f \in E(G)$ of $C$.
\end{definition}
Note that any $M$-alternating cycle $C$ (both simple modifiers and cross modifiers are $M$-alternating cycles) can be used to find a new PM $M' := M \symdif C$. In such a case we usually say that $M'$ is obtained from $M$ by \emph{switching $M$ along $C$}.
\begin{figure}[htpb]
    \centering
    \includegraphics[scale=0.85]{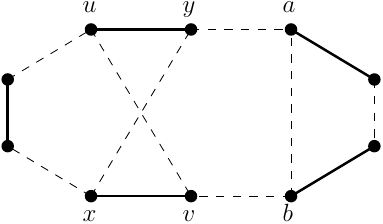}
    \caption{An example of an $M$-alternating $10$-cycle $C$ with three chords. Edges not from $M$ are dashed. The odd chord $\{a, b\}$ induces a simple modifier $C'$ of $C$ of length $4$. The two even chords $\{u, v\}$ and $\{x, y\}$ induce a cross modifier $C''$ of $C$ of length $4$. Observe that $C'$ and $C''$ are disjoint and hence $M_1 := M \symdif C'$, $M_2 := M \symdif C''$ and $M_3 := M \symdif C' \symdif C''$ are all valid PMs that can be obtained by switching $M$ along $C'$, $C''$, or both. }
    \label{fig:modifiers}
\end{figure}
\begin{figure}[htpb]
    \centering
    \includegraphics[scale=0.85]{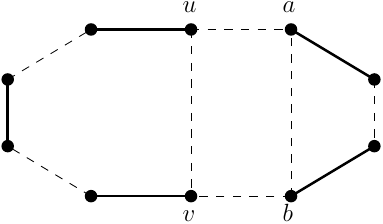}
    \caption{An example of an $M$-alternating $10$-cycle with two adjacent odd chords $\{u, v\}$ and $\{a, b\}$ that both induce an $M$-alternating simple modifier. The two modifiers are disjoint and hence this is exactly the situation from Lemma~\ref{lemma:properties_modifiers}.}
    \label{fig:disjoint_modifiers}
\end{figure}

We are now ready to introduce the rank of a symmetric difference. Intuitively, a smaller rank typically means that the symmetric difference consists of fewer edges, multiple cycles, and that it admits more disjoint simple modifiers. In particular, the definition prioritizes reducing the number of edges $|M \symdif M'|$ over increasing the number of cycles. For equal number of edges and cycles, we prefer having more disjoint simple modifiers. One can think of this as a lexicographic ordering of these three attributes, but we decided to put it all together into one number according to the following definition.
\begin{definition}[Rank]
\label{definition:rank}
    Let $G$ be a graph with PMs $M$ and $M'$. Let $z$ be the number of cycles in $M \symdif M'$. Moreover, let $z_M$ denote the maximum number of pairwise disjoint $M$-alternating simple modifiers with respect to cycles from $M \symdif M'$. Analogously, let $z_{M'}$ denote the maximum number of pairwise disjoint $M'$-alternating simple modifiers with respect to cycles from $M \symdif M'$. Then the rank of $M \symdif M'$ is given by
    \[
        \rank(M \symdif M') := |M \symdif M'| + \frac{1}{1 + z} + \frac{1}{n (1 + \max \{z_{M}, z_{M'} \})}
    \]
    where $n := |V|$.
\end{definition}
These newly introduced concepts can now be used to prove Lemma~\ref{lemma:properties_modifiers}. This can be seen as an introduction to the type of arguments that we will be making throughout this subsection.

\begin{lemma}
\label{lemma:properties_modifiers}
    Let $G$ be a graph with PMs $M$ and $M'$ and let $C$ be a cycle in $M \symdif M'$. Moreover, let $C'$ and $C''$ be two disjoint $M$-alternating simple modifiers of $C$. Consider the three PMs $M_1 := M \symdif C'$, $M_2 := M \symdif C''$, and $M_3 := M \symdif C' \symdif C''$.
    Then we have
    \[
        \max\{\rank(M_i \symdif M), \rank(M_i \symdif M') \} < \rank(M \symdif M')
    \]
    for all $i \in \{1, 2, 3\}$. 
\end{lemma}
\begin{proof}
    Consider first the case of $M_1$. As any simple modifier must have cardinality at least $4$ but cannot be equal to all of $C$, we observe that both $M_1 \symdif M$ as well as $M_1 \symdif M'$ must have fewer edges than $M \symdif M'$ and hence we get 
    \[
        \max\{\rank(M_1 \symdif M), \rank(M_1 \symdif M') \} < \rank(M \symdif M')
    \]
    as desired. This works for $M_2$ analogously.

    It remains to consider $M_3$. 
    Observe first that we have $|M_3 \symdif M| \leq |M \symdif M'|$ and $|M_3 \symdif M'| \leq |M \symdif M'|$, i.e.\ the number of edges cannot increase in either case. In fact, we do have $|M_3 \symdif M'| < |M \symdif M'|$ and consequently $\rank(M_3 \symdif M') < \rank(M \symdif M')$. However, in the case of $C$ being the only cycle in $M \symdif M'$ and $|C \symdif C' \symdif C''| = 4$, we might have $|M \symdif M_3| = |M \symdif M'|$. Thus, an additional argument is required here. 
    Observe that both modifiers $C'$ and $C''$ will become cycles of $M \symdif M_3$. In particular, they replace the cycle $C$. Consequently, the number of cycles in $M \symdif M_3$ is strictly larger than in $M \symdif M'$, i.e.\ the rank is smaller.
\end{proof}
Recall that for Karzanov's property, we are interested in maintaining the parity of the number of red edges of our PMs. Lemma~\ref{lemma:properties_modifiers} is interesting because no matter the coloring, at least one of $M_1, M_2, M_3$ will have the same parity as $M$. This type of argument will reappear in the proof of Lemma~\ref{lemma:rank_reduction} which can be considered the core argument of this subsection. In particular, the proof of Lemma~\ref{lemma:rank_reduction} consists of analyzing various cases that can occur and Lemma~\ref{lemma:properties_modifiers} deals with one of those cases.

\begin{lemma}
\label{lemma:rank_reduction}
    Let $(G, c)$ be a colored graph such that $G$ satisfies the chord property. Moreover, assume that there exist PMs $M$ and $M'$ in $G$ satisfying $r(M) \equiv_2 r(M')$ and $r(M) + 4 \leq r(M')$. Then there is 
    a PM $M''$ in $G$ with $r(M'') \equiv_2 r(M)$ and 
    \[
        \max\{\rank(M'' \symdif M), \rank(M'' \symdif M') \} < \rank(M \symdif M').
    \]
\end{lemma}
\begin{proof}
    Note that our assumption $r(M) + 4 \leq r(M')$ implies $|M \symdif M'| \geq 8$.
    Consider the symmetric difference $M \symdif M'$. We distinguish three cases based on the number of cycles in  $M \symdif M'$. 
    
    \textbf{Case 1:} Assume that $M \symdif M'$
    consists of at least three cycles $C, C', C''$. Then $M''$ can be found by switching $M$ along $C$, $C'$ or both. In all three cases, we have 
    \[
        \max\{\rank(M'' \symdif M), \rank(M'' \symdif M') \} < \rank(M \symdif M').
    \]
    as the number of edges in the symmetric difference decreases.

    \textbf{Case 2:} Now assume instead that $M \symdif M'$ consists of exactly two cycles 
    $C$ and $C'$. If they both have length $4$, we must have $r(M') = r(M) + 4$ and switching $M$ along one of the cycles will yield the desired $M''$. Thus, assume that $C$ has length at least $6$. Moreover, assume first that $C$ has a simple modifier $C_e$ induced by the odd chord $e$. This would allow us to switch $M$ along $C', C_e$ or both to obtain $M''$. Again, the decrease in rank is guaranteed by a reduction of the number of edges. If $C$ does not contain an odd chord, the chord property implies the existence of all even chords of $C$. In particular, there must be an $M$-alternating cross modifier $C_{e, f}$ of length $4$ induced by the even chords $e, f$ in $C$. We would like to switch $M$ along $C', C_{e, f}$ or both and guarantee that the rank decreases in all three cases. It is clear that the number of edges is reduced if we only use $C'$ or if we use both $C'$ and $C_{e, f}$. However, the number of edges and cycles of $M'' \symdif M'$ is the same as in $M \symdif M'$ in the case of $M'' := M \symdif C_{e, f}$. Hence, we need to consider the number of pairwise disjoint simple modifiers. Notice that there are none in $C$ and that $C$ is replaced by the cycle $C \symdif C_{e, f}$ in $M'' \symdif M'$. In particular, some of the even chords of $C$ are odd chords of the cycle $C \symdif C_{e, f}$. Consequently, the maximum number of disjoint simple modifiers of the same alternation increases and hence the rank decreases,
    as desired.

    \textbf{Case 3:} Finally, assume that $M \symdif M'$ consists of a single cycle $C$. 
    Observe that $C$ must have length at least $8$. Assume first that $C$ admits two disjoint simple modifiers $C_1, C_2$ of the same alternation. Without loss of generality, assume that they are both $M$-alternating. Switching $M$ along $C_1, C_2$ or both will yield the desired $M''$. Moreover, Lemma~\ref{lemma:properties_modifiers} guarantees a decrease of the rank in all cases.

    Thus, assume now that there are no two disjoint simple modifiers of the same alternation and let $v_0, v_1, \dots, v_7$ be $8$ consecutive vertices on $C$. Since we just excluded the existence of two adjacent odd chords, the chord property guarantees the existence of the even chords $\{v_i, v_{(i + 4) \, \% \, 8}\} \in E$ for all $0 \leq i \leq 7$. This implies two disjoint 
    $M$-alternating cross modifiers $C'$ and $C''$. Without loss of generality, assume that $C'$ is induced by the two chords $\{v_1, v_5 \}$ and $\{v_2, v_6\}$ and $C''$ is induced by $\{v_0, v_4 \}$ and $\{v_3, v_7\}$ (see Figure~\ref{fig:many_even_chords}).
    Consequently, we can switch $M$ along $C', C''$ or both to obtain $M''$. However, it is not obvious that 
    \[
        \max\{\rank(M'' \symdif M), \rank(M'' \symdif M') \} < \rank(M \symdif M')
    \]
    holds in all three cases: Although we will always have $\rank(M'' \symdif M) < \rank(M \symdif M')$
    because the number of edges decreases or the number of cycles increases, the same argument 
    might fail for $\rank(M'' \symdif M')$. This happens exactly if only one of $C'$
    and $C''$ is used (and the used modifier has length $4$). In this case, 
    $M'' \symdif M'$ will contain the same number of edges and cycles as $M \symdif M'$.
    But we know that $M \symdif M'$ did not contain two disjoint simple modifiers of the same alternation (otherwise 
    we would have used those). Moreover, using only $C'$ will make the edges $\{v_0, v_4 \}$ and $\{v_3, v_7\}$ odd chords such that they induce two disjoint $M'$-alternating simple modifiers. Similarly, using only $C''$ will make $\{v_1, v_5 \}$ and $\{v_2, v_6\}$ odd chords that induce disjoint simple modifiers of the same alternation. Hence, $M \symdif M''$ must contain two new
    disjoint simple modifiers of the same alternation in both cases. This yields the desired decrease 
    in rank. 
\end{proof}
\begin{figure}[htpb]
    \centering
    \includegraphics[scale=0.85]{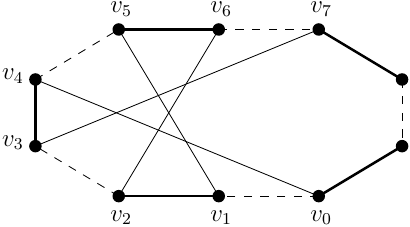}
    \caption{A sketch of Case $3$ from the proof of Lemma~\ref{lemma:rank_reduction}. 
    The shown $10$-cycle is alternating in $M$ (bold edges) and $M'$ (dashed edges). Moreover, the four even chords from the proof are shown. The chords $\{v_1, v_5 \}$ and $\{v_2, v_6\}$ induce an $M$-alternating cross modifier of length $4$. The chords $\{v_0, v_4 \}$ and $\{v_3, v_7\}$ induce an $M$-alternating cross modifier of length $6$. The two cross modifiers are disjoint.}
    \label{fig:many_even_chords}
\end{figure}
It remains to put all things together and conclude with a proof of Lemma~\ref{lemma:chord_property_is_sufficient}. We already gave many of the remaining ideas away at the beginning of this subsection.

\chordproperty*
\begin{proof}
    Fix an integer $k$ 
    such that there exist PMs
    $M$ and $M'$ with $r(M) \leq k \leq r(M')$ and $k \equiv_2 r(M) \equiv_2 r(M')$ in $(G, c)$.
    Moreover, choose $M$ and $M'$ according to those constraints such that additionally, $\rank(M \symdif M')$ is minimized.
    If we have $r(M') < r(M) + 4$, then either $r(M') = k$ or $r(M) = k$ and we are done. Thus, assume now that we have $r(M') \geq r(M) + 4$. 
    Lemma~\ref{lemma:rank_reduction} guarantees the existence of $M''$ with $r(M'') \equiv_2 r(M)$ and 
    \[
        \max\{\rank(M'' \symdif M), \rank(M'' \symdif M') \} < \rank(M \symdif M').
    \]
    Consequently, we get a contradiction with the minimality of $\rank(M \symdif M')$ as $M''$ can either replace $M$ or $M'$ depending on $r(M'')$.
    It follows that $r(M') \geq r(M) + 4$ is impossible.
\end{proof}

\subsection{Chord Property Implies an Algorithm for BCPM} 
\label{subsection:proof_lemma_bcpm_algo}

In this subsection, we give a deterministic polynomial-time algorithm for \BCPM\ on graphs that have the chord property. This proves Lemma~\ref{lemma:polynomial-time_algorithm_bcpm} from Section~\ref{sec:karzanov}. In particular, we claim that the simple 
Algorithm~\ref{alg:bcpm} is always successful on graphs that satisfy the chord property.
\begin{algorithm}
\KwData{Colored graph $(G, c)$ with $G$ satisfying the chord property, integer $k \geq 0$}
\KwResult{A PM $M$ with $r(M) \equiv_2 k$ and $r(M) \leq k$, or the information that there is no such PM}

$\mathcal{F} \gets \{ F \subseteq E(G) \mid |F| \leq 4 \text{ and no two edges of } F \text{ intersect} \}$\;

\For{$F \in \mathcal{F}$}{
    $G'$ is obtained by deleting all vertices that are covered by $F$ from $G$\;
    \If{$G'$ admits a PM}{
        $\Mmin' \gets $ PM in $G'$ with minimum number of red edges among all PMs\;
        $M \gets \Mmin' \cup F$\;
        \If{$r(M) \leq k$ and $r(M) \equiv_2 k$}{
        	return $M$\;
        }
    }  
}
return that there is no such PM\;

\caption{An algorithm for \BCPM\ on graphs with the chord property}
\label{alg:bcpm}
\end{algorithm}
While it is clear that Algorithm~\ref{alg:bcpm} runs in deterministic polynomial-time, the main challenge of this subsection is to prove its correctness. Our proof roughly works as follows: Consider first the PM $\Mmin$ with the fewest red edges in a given colored graph $(G, c)$. If we have $r(\Mmin) > k$, then clearly Algorithm~\ref{alg:bcpm} will correctly return that there is no PM that satisfies the constraints. Moreover, note that Algorithm~\ref{alg:bcpm} will eventually inspect a PM with the same number of red edges as $\Mmin$. Hence, the algorithm is also correct in all cases where $r(\Mmin) \equiv_2 k$. Thus it remains to consider the case where we have $r(\Mmin) < k$ and $r(M) \nequiv_2 k$. Assume additionally that there is a solution PM $M^*$ with $r(M^*) \leq k$ and $r(M^*) \equiv_2 k$. We will show in Corollary~\ref{corollary:chord_property_bcpm_preparation} that under these assumptions, there must be two PMs $M$ and $M'$ with $r(M) \nequiv_2 r(M')$, $r(M) \leq k + 1$, $r(M') \leq k + 1$, and $|M \symdif M'| \leq 8$. 
This then implies that Algorithm~\ref{alg:bcpm} will correctly output the desired PM. The general idea behind this is that the algorithm will eventually inspect $F = M \cap (M \symdif M')$ and $F' = M' \cap (M \symdif M')$. In one of the two cases it will terminate and output the PM with the desired parity since the $M$ and $M'$ have different parities. The detailed argument can be found in the proof of Lemma~\ref{lemma:polynomial-time_algorithm_bcpm} at the end of this subsection.
Corollary~\ref{corollary:chord_property_bcpm_preparation} follows from Lemma~\ref{lemma:chord_property_bcpm_preparation}. Roughly speaking, Lemma~\ref{lemma:chord_property_bcpm_preparation} allows us to reduce the rank of symmetric difference between two PMs while increasing the number of red edges by at most one. Moreover, we can avoid this increase in the number of red edges with an additional assumption. This is crucial for applying Lemma~\ref{lemma:chord_property_bcpm_preparation} in the proof of Corollary~\ref{corollary:chord_property_bcpm_preparation}. 

While our arguments are similar to the ones from Subsection~\ref{subsection:proof_lemma_chord_property}, there are two major differences: Firstly, in Subsection~\ref{subsection:proof_lemma_chord_property} we usually cared about having two PMs with the same parity of red edges as $k$. Here, we instead care about having two PMs with opposite parity of red edges (see Corollary~\ref{corollary:chord_property_bcpm_preparation}). Secondly, in Subsection~\ref{subsection:proof_lemma_chord_property} we wanted to have PMs $M$ and $M'$ with $r(M) \leq k \leq r(M')$. Here, we instead want that our two PMs $M$ and $M'$ both have at most $k + 1$ edges (at most $k$ would be desirable but we can make it work with at most $k + 1$).

\begin{lemma}
\label{lemma:chord_property_bcpm_preparation}
    Let $(G, c)$ be a colored graph such that $G$ satisfies the chord property. Consider two 
    PMs $M$ and $M'$ with $r(M) < r(M')$. If we have $\rank(M \symdif M') > 8$, then there exists a PM $M''$ with $r(M'') \leq r(M') + 1$ and 
    \[
        \max \{ \rank(M'' \symdif M), \rank(M'' \symdif M') \} < \rank(M \symdif M').
    \]
    If additionally, we are guaranteed that there is no PM $M'''$ with $r(M''') = r(M') + 1$ and $|M' \symdif M'''| = 4$, then we get the better guarantee $r(M'') \leq r(M')$.
\end{lemma}
\begin{proof}
    We will distinguish three cases.
    
    \textbf{Case 1:} Assume that $M \symdif M'$ contains more than one cycle. Then there must be some cycle $C$ such that $r(M \symdif C) \leq r(M')$ and hence we can find $M'' := M \symdif C$ with $r(M'') \leq r(M')$ accordingly.
    
    \textbf{Case 2:} Next, assume that $M \symdif M'$ consists of a single cycle $C$ that has two adjacent odd chords $e$ and $f$.
    Denote by $C_e$ and $C_f$ the corresponding disjoint simple modifiers of the same alternation. Since $e$ and $f$ are adjacent, we have that $|C \symdif C_e \symdif C_f| = 4$.

    Assume first that $C_e$ and $C_f$ are $M$-alternating.
    From $|C \symdif C_e \symdif C_f| = 4$ we get the upper bound $r(M') + 2 \geq r(M \symdif C_e \symdif C_f)$. Combining this with the observation 
    \[
        r(M \symdif C_e \symdif C_f) = r(M \symdif C_e) + r(M \symdif C_f) - r(M)
    \]
    yields 
    \[
        r(M') + r(M) + 2 \geq r(M \symdif C_e) + r(M \symdif C_f).
    \]
    By $r(M) < r(M')$, we get 
    \[
        2r(M') + 1 \geq r(M \symdif C_e) + r(M \symdif C_f)
    \]
    and hence we either have $r(M') \geq r(M \symdif C_e)$ or $r(M') \geq r(M \symdif C_f)$.
    Choosing $M'' := M \symdif C_e$ or $M'' := M \symdif C_f$
    accordingly concludes the argument. 

    Consider now the case where $C_e$ and $C_f$ are $M'$-alternating. Consider the 
    $M$-alternating cycle $C' := C \symdif C_e \symdif C_f$ of length $4$. We have $r(M \symdif C') \leq r(M) + 2$ and hence are done if $r(M) \leq r(M') - 2$.
    Thus, assume now that $r(M) = r(M') - 1$ and that $r(M \symdif C') = r(M) + 2 = r(M') + 1$.
    Observe that $M' \symdif C_e \symdif C_f = M \symdif C'$ and therefore, we have 
    \[
        r(M') + 1 = r(M' \symdif C_e \symdif C_f) = r(M' \symdif C_e)
        + r(M' \symdif C_f) - r(M').
    \] 
    We can rewrite this to 
    \[
        2r(M') + 1 = r(M' \symdif C_e)
    + r(M' \symdif C_f)
    \] 
    and hence either $M' \symdif C_e$ or $M' \symdif C_f$ has at most $r(M')$ many red edges. 
    Again, choosing $M'' := M' \symdif C_e$ or $M'' := M' \symdif C_f$ accordingly concludes the argument.
    
    \textbf{Case 3:} It remains to consider the case where $M \symdif M'$ consists of a single cycle $C$ that does not have a pair of adjacent odd chords. 
    Consider eight consecutive vertices $v_0, \dots, v_7$ on $C$. In particular, we choose them such that $\{v_0, v_1\} \in C \cap M'$ is a red edge.
    Such a red edge must exist in $C \cap M'$ since we have $r(M) < r(M')$. The chord property now guarantees the existence of the chords $\{v_i, v_{(i + 4) \, \% \, 8}\} \in E$ for all $0 \leq i \leq 7$. In particular, the chords $\{v_0, v_4\}$ and $\{v_1, v_5\}$ induce an $M'$-alternating cross modifier $C'$ of length $4$ that includes the red edge $\{v_0, v_1\}$. Consider $M'' := M' \symdif C'$. 
    Since $\{v_0, v_1\}$ is red, we must have $r(M'') \leq r(M') + 1$, as desired.
    Our additional assumption in the statement of the lemma rules out the case $r(M'') = r(M') + 1$. Thus, under this additional assumption we get $r(M'') \leq r(M')$. It remains to observe that we also have 
    \[
        \max\{\rank(M'' \symdif M), \rank(M'' \symdif M') \} < \rank(M \symdif M').
    \]
    In particular, the inequality $\rank(M'' \symdif M) < \rank(M \symdif M')$ 
    follows by noticing that the even chords $\{v_2, v_6\}$ and $\{v_3, v_7\}$ will become odd chords in $M'' \symdif M$ that induce simple modifiers of the same alternation. Hence, the maximum number of disjoint 
    simple modifiers increases.
\end{proof}
The main technical part is now behind us. It remains to apply Lemma~\ref{lemma:chord_property_bcpm_preparation} to prove the following corollary. This is the main ingredient that we need for the correctness of Algorithm~\ref{alg:bcpm}.
\begin{corollary}
\label{corollary:chord_property_bcpm_preparation}
    Let $(G, c)$ be a colored graph such that $G$ satisfies the chord property. Moreover, let $k$ be an integer such that there exist PMs of both parities with at most $k$ red edges in $G$. Then there are PMs $M$ and $M'$ with 
    $r(M) \leq k + 1$, $r(M') \leq k + 1$, $r(M) \nequiv_2 r(M')$, and $|M \symdif M'| \leq 8$.
\end{corollary}
\begin{proof}
    Note that we can assume that there are no two PMs $M, M'$
    with $r(M) = k$, $r(M') = k + 1$ and $|M \symdif M'| = 4$, as otherwise we would be done already. Let now $M$ and $M'$ be PMs satisfying $r(M) \nequiv_2 r(M')$
    and $\max \{r(M), r(M') \} \leq k$ such that $\rank(M \symdif M')$ is minimized. We prove that we must have $|M \symdif M'| \leq 8$. Assume for a contradiction that we instead have $|M \symdif M'| > 8$ and consequently $\rank(M \symdif M') > 8$. 

    \textbf{Case 1:} Consider first the case where we have 
    $
        \max \{r(M), r(M') \} \leq k - 1.
    $ 
    Then Lemma~\ref{lemma:chord_property_bcpm_preparation} implies the existence of a new perfect matching $M''$ with at most $r(M') + 1 \leq k$ red edges that replaces either $M$ or $M'$ (depending on its parity). In particular, we have 
    \[
        \max\{\rank(M'' \symdif M), \rank(M'' \symdif M') \} < \rank(M \symdif M')
    \]
    and hence a contradiction with our assumption that $\rank(M \symdif M')$ is minimal.

    \textbf{Case 2:} It remains to consider the case where we have $\max \{r(M), r(M') \} = k$.
    Without loss of generality, assume that we have $r(M) < r(M') = k$. Note that in this case, the additional assumption of Lemma~\ref{lemma:chord_property_bcpm_preparation} is satisfied. Hence, we can again conclude existence of a new perfect matching $M''$ with at most $k$ red edges that either replaces $M$ or $M'$ and leads to the desired contradiction.  
\end{proof}
We conclude this subsection with the proof of Lemma~\ref{lemma:polynomial-time_algorithm_bcpm}. In particular, the proof uses the ideas that we outlined in the beginning of this subsection in order to show that Algorithm~\ref{alg:bcpm} is correct.
\polytimebcpm*
\begin{proof}
    It is clear that Algorithm~\ref{alg:bcpm} runs in deterministic polynomial-time. Hence, it remains to prove its correctness. 
    To that end, assume the algorithm is given a colored graph $(G, c)$ and integer $k$ as input and let $\Mmin$ be a PM of $(G, c)$ with as few red edges as possible. 

    It is not hard to see that the algorithm is correct if we have $r(\Mmin) > k$ or $r(\Mmin) \equiv_2 k$. Thus assume now that we have $r(\Mmin) \leq k$ and $r(\Mmin) \nequiv_2 k$. Moreover, assume that there is a solution PM $M^*$ with $r(M^*) \leq k$ and $r(M^*) \equiv_2 k$. In particular, the assumption of Corollary~\ref{corollary:chord_property_bcpm_preparation} is satisfied and we conclude that there must be PMs $M_1$ and $M_2$
    with $r(M_1) \leq k + 1$, $r(M_2) \leq k + 1$, $r(M_1) \nequiv_2 r(M_2)$, and $|M_1 \symdif M_2| \leq 8$. Both $M_1$ and $M_2$ have at most four edges in $M_1 \symdif M_2$. Hence, Algorithm~\ref{alg:bcpm} will eventually consider the sets of edges $F_1 = M_1 \cap (M_1 \symdif M_2)$ and $F_2 = M_2 \cap (M_1 \symdif M_2)$. Without loss of generality, assume that we have $r(M_1) \equiv_2 k$ and therefore also $r(M_1) \leq k$. We claim that Algorithm~\ref{alg:bcpm} will terminate after considering $F_1$ or $F_2$. Assume it does not terminate after considering $F_1$. Recall that the algorithm computes the PM $\Mmin'$ on the reduced graph $G'$ obtained from $G$ by deleting the vertices covered by $F_1$. The set $F_1 \cup \Mmin'$ is a PM of $G$ and by minimality of $\Mmin'$, it has at most $r(M_1) \leq k$ red edges. Since we assume that the algorithm does not terminate, we therefore must have $r(F_1 \cup \Mmin') \nequiv_2 k$ and $r(F_1 \cup \Mmin') < k$. But then we have $r(F_2 \cup \Mmin') \equiv_2 k$. Using minimality of $\Mmin'$ again and the fact $r(M_2) \leq k + 1$, we get that we must have $r(F_2 \cup \Mmin') \leq k$. Therefore, the algorithm terminates as soon as it considers $F_2$ since $F_1$ and $F_2$ cover the exact same set of vertices.
\end{proof}

\subsection{Chain Graphs have the Chord Property}
\label{subsection:chain_graph}

In this subsection we will prove that all chain graphs satisfy the chord property. It is well-known and not hard to see that every chain graph is also bipartite chordal, i.e.\ every cycle of length at least $6$ in a chain graph has an odd chord. Together with Lemma~\ref{lemma:double_cut_in_chain_graph}, this implies that all chain graphs satisfy the chord property.

\begin{lemma}
\label{lemma:double_cut_in_chain_graph}
    Let $G$ be a chain graph. Then every cycle 
    $C$ in $G$ of (necessarily even) length at least $8$ has two adjacent odd chords. 
\end{lemma}
\begin{proof}
    Let $C$ be an arbitrary cycle of length at least $8$ and let $X, Y$ denote the partitions of $G$, i.e.\ $V(G) = X \dotunion Y$. Moreover, assume that vertices are labelled as follows: 
    $x_1, \dots, x_{|X|} \in X$ and $y_1, \dots, y_{|Y|} \in Y$
    such that $N(x_i) \subseteq N(x_{i + 1})$ and $N(y_j) \subseteq N(y_{j + 1})$ 
    hold for all $1 \leq i < |X| $ and $1 \leq j < |Y|$. Since we have $|C| \geq 8$, there exist four distinct vertices $x_a, x_b, x_c, x_d \in X \cap C$. Without loss of generality, assume that they appear in this order $x_a, x_b, x_c, x_d$ on $C$, that we have $a < \min \{b, c, d \}$, and that $x_d$ and $x_b$ are at distance exactly $2$ from $x_a$ along $C$.
    In particular, we have $N(x_a) \subseteq N(x_b)$ and $N(x_a) \subseteq N(x_d)$. Moreover, assume that we have $b < d$ and hence $N(x_b) \subseteq N(x_d)$ (the case $N(x_d) \subseteq N(x_b)$ is symmetric). 
    A sketch of this situation is given in Figure~\ref{figure:chain_graph}.

    We will now observe that two adjacent chords have to exist in $C$. For the first chord, let $y^{(a, d)}$ be the common neighbor of $x_a$ and $x_d$ on $C$.
    By $N(x_a) \subseteq N(x_b)$, the edge $\{x_b, y^{(a, d)} \}$ has to exist. For our second chord, let $y^{(b)}$ be the neighbor of $x_b$ which is not adjacent to $x_a$ on $C$. By $N(x_b) \subseteq N(x_d)$, the edge $\{x_d, y^{(b)}\}$ has to exist. We hence found the two adjacent odd chords $\{x_b, y^{(a, d)}\}$ and $\{x_d, y^{(b)}\}$. 
\end{proof}

\begin{figure}[htpb]
    \centering
    \includegraphics[scale=0.85]{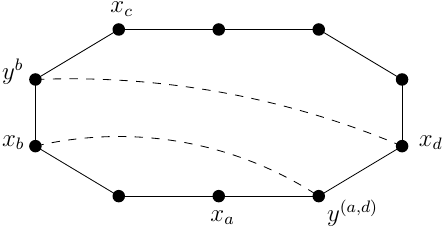}
    \caption{A sketch of the situation in the proof of Lemma~\ref{lemma:double_cut_in_chain_graph}. In the proof we argue that the 
    two adjacent odd chords shown by dashed edges have to exist. }
    \label{figure:chain_graph}
\end{figure}

\begin{corollary}
    Every chain graph satisfies the chord property.
\end{corollary}

\subsection{Unit Interval Graphs have the Chord Property} 

In this subsection, we will prove that intersection graphs of unit intervals (i.e.\ unit interval graphs) satisfy the chord property. From our framework in Section~\ref{sec:karzanov} it then follows that \EM\ restricted to unit interval graphs can be decided in deterministic polynomial-time. 

It is well-known that unit interval graphs are strongly chordal. Hence, every even cycle of length at least $6$ in a unit interval graph contains an odd chord. Together with Lemma~\ref{lemma:double-cut_unit-interval} below, this implies that every unit interval graph satisfies the chord property. 

\begin{lemma}
\label{lemma:double-cut_unit-interval}
    Let $G$ be a unit interval graph. Then every cycle $C$ of even length at least 
    $8$ in $G$ contains two adjacent odd chords.
\end{lemma}
\begin{proof}
    Let $C$ be an arbitrary cycle of even length at least $8$ and 
    let $I$ be a unit interval representation of $G$, i.e.\ $I(v)$ is the unit-length 
    interval associated with $v \in V(G)$. We introduce 
    the following notation for two unit intervals $A = [l_A, r_A]$ and $B = [l_B, r_B]$: 
    We write $A \leq B$ if we have $l_A \leq l_B$ (i.e.\ $A$ lies to the left of $B$ but they might intersect). 
    Moreover, we use $A < B$ to additionally imply that $r_A < l_B$ and hence $A \cap B = \emptyset$. 
    The relations $>, \geq$ are defined analogously. In the special case of $A = B$ 
    we will say that both $A \leq B$ and $B \leq A$ are satisfied.

    We now get to the key part of the proof. 
    Consider two distinct vertices $u, v$ on $C$ such that the shorter path 
    between $u, v$ along $C$ contains exactly $4$ edges. In particular, $u$ and $v$ are an 
    even distance apart (on $C$). Without loss of generality, assume that we 
    have $I(u) \leq I(v)$. Next, fix some arbitrary 
    orientation of $C$ and consider two walkers that walk along $C$ 
    according to the chosen orientation. One of them starts at $u$ while the other 
    one starts at $v$. In every step of the procedure, both walkers take one step along $C$ simultaneously. In particular, the distance between the two walkers along the shorter path on $C$ always remains exactly $4$. 
    Let now $x$ and $y$ be the positions of the two walkers when we have $I(x) \geq I(y)$ for the first time ($x$ is the position of the walker that started at $u$ and $y$ is the position of the walker that started at $v$). Note that this must happen eventually because the walker starting at $u$ will reach the right-most interval at some point. Moreover, note that we require the walkers to take at least one step as otherwise, we might run into a problem if $I(u) = I(v)$. In other words, we have $x \neq u$ and $y \neq v$.
    
    Now consider the vertices $p_x$ and $p_y$ that are predecessors of $x$ and $y$ along $C$ according to the chosen orientation. By definition of $x$ and $y$, we must have $I(p_x) \leq I(p_y)$. 
    Observe that $I(p_x)$ and $I(x)$ must intersect. Hence, if $I(p_x)$ does not intersect $I(y)$, we must have $I(p_x) > I(y)$. But this contradicts 
    $I(p_x) \leq I(p_y)$ since $I(p_y)$ must intersect $I(y)$. Hence, $I(p_x)$ must intersect $I(y)$. 
    Analogously, $I(p_y)$ must intersect $I(x)$. The two chords  
    $\{x, p_y\}, \{y, p_x\}$ are adjacent and both odd, as desired. This situation is also sketched in Figure~\ref{figure:unit_interval_chords}. 
\end{proof}
\begin{figure}[htpb]
    \centering
    \includegraphics[scale=0.85]{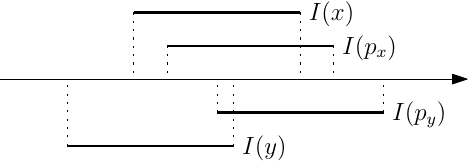}
    \caption{This figure shows the situation at the end of the proof of Lemma~\ref{lemma:double-cut_unit-interval}. In particular, the four unit intervals from the proof are shown. Note that only the horizontal position of the intervals matters, the vertical position is irrelevant. In the proof we argue that $I(x)$ must intersect $I(p_y)$ and $I(y)$ must intersect $I(p_x)$. }
    \label{figure:unit_interval_chords}
\end{figure}
Observe that we only have to consider odd chords here. In particular, the proof of Lemma~\ref{lemma:double-cut_unit-interval} also works in bipartite unit interval graphs.
Therefore, we conclude this subsection with the following result
\begin{corollary}
    All unit interval graphs and bipartite unit interval graphs satisfy the chord property.
\end{corollary}

\subsection[Complete r-Partite Graphs have the Chord Property]{Complete $r$-Partite Graphs have the Chord Property}
\label{subsection:complete_r-partite}

In this subsection, we prove that all complete $r$-partite graphs satisfy the chord property. While it was sufficient to work with odd chords in the case of unit interval graphs and chain graphs, we will have to also consider even chords here. We verify that the two conditions of the chord property hold in Lemma~\ref{lemma:r_partite_6_cycles} and Lemma~\ref{lemma:r_partite_8_cycles}, respectively.

\begin{lemma}
\label{lemma:r_partite_6_cycles}
    Let $C$ be a cycle of even length at least $6$ in a complete $r$-partite graph $G$. If $C$ 
    does not have an odd chord, then $C$ must have all even chords.
\end{lemma}
\begin{proof}
    We provide an indirect proof. Assume that there are two vertices $u, v$ 
    of even distance on $C$ such that $\{u, v\} \notin E(G)$. Since $C$ has length at least $6$, there must be a vertex $w$ adjacent (on $C$) to $u$ but not to $v$. By $\{u, v\} \notin E(G)$, we know that $u$ and $v$ belong to the same partition of $G$. Furthermore, $\{u, w\} \in E(G)$
    tells us that $w$ is in a different partition than $u$ and $v$. Consequently, the odd chord 
    $\{v, w\}$ must exist in $G$.
\end{proof} 

\begin{lemma}
\label{lemma:r_partite_8_cycles}
    Let $C$ be a cycle of even length at least $8$ in a complete $r$-partite graph $G$.
    Assume that $C$ does not contain two adjacent odd chords.
    Then all even chords of $C$ with split at least $4$ exist. 
\end{lemma}
\begin{proof}
    We will proceed indirectly. Assume that there exist vertices $u, v$ on $C$ with $\{u, v\} \notin E$ such that the distance (along $C$) between $u$ and $v$ is even and at least $4$. We choose an orientation of $C$ arbitrarily. Let now $s_u$ be the successor of $u$ along $C$ and let $s_v$ be the successor of $v$ along $C$ according to the chosen orientation. From $\{u, v\} \notin E$ it follows that $u$ and $v$ must belong to the same partition of $G$. But neither of $s_u$ and $s_v$ can belong to this partition. Hence, the edges $\{u, s_v\}$ and $\{v, s_u\}$ exist in $G$. These two edges are adjacent odd chords of $C$.
\end{proof}

\begin{corollary}
    Every complete $r$-partite graph satisfies the chord property.
\end{corollary}

\subsection{Karzanov's Weak Property in Bipartite Chordal and Strongly Chordal Graphs}
\label{subsec:weak_karzanov}

In this subsection, we prove that all colored bipartite chordal and strongly chordal graphs satisfy Karzanov's weak property. The proof mainly exploits that even cycles of length at least $6$ always have an odd chord. It is significantly easier than the proofs in Subsections~\ref{subsection:proof_lemma_chord_property} and~\ref{subsection:proof_lemma_bcpm_algo}.
\weakkarzanov*
\begin{proof}
    Let $(G, c)$ be such a colored graph and fix an integer $k$ such that there are PMs $M, M'$ with $r(M) \leq k \leq r(M')$.
    Moreover, assume that $M$ and $M'$ are such that $\rank(M \symdif M')$ is minimized.
    Observe that $\rank(M \symdif M') \leq 6$ implies $r(M') \leq r(M) + 2$. Hence, the statement clearly holds in this case.

    Thus, assume now that we have $\rank(M \symdif M') > 6$. If $M \symdif M'$ contains at least two disjoint cycles $C$ and $C'$, it suffices to consider $M'' := M \symdif C$.
    We have 
    \[
        \max\{\rank(M'' \symdif M), \rank(M'' \symdif M') \} < \rank(M \symdif M')
    \]
    and hence we can replace either $M$ or $M'$ (depending on $r(M'')$) with $M''$ to get 
    a contradiction to the minimality of $\rank(M \symdif M')$.

    It remains to handle the case where $M \symdif M'$ forms a single cycle $C$. 
    From $\rank(M \symdif M') > 6$ it follows that this cycle must have even 
    length at least $6$ and hence it admits an odd chord $e$. The chord $e$ induces 
    an $M$-alternating simple modifier $C_e$, and we define $M'' := M \symdif C_e$.
    Again, we have 
    \[
        \max\{\rank(M'' \symdif M), \rank(M'' \symdif M') \} < \rank(M \symdif M')
    \]
    and we get the desired contradiction.
\end{proof}
\end{document}